\documentclass[12pt]{article} 
\usepackage[margin=1in]{geometry} 
\usepackage[sectionbib]{natbib}
\usepackage{chapterbib}
\usepackage{graphicx}
\usepackage{hyperref}  

\usepackage{amsmath}
\usepackage{amssymb}
\usepackage{amsfonts}
\usepackage{multirow}
\usepackage{amsthm}
\usepackage{xcolor}
\usepackage{bm}

\usepackage{A_command}

\newtheorem{theorem}{Theorem}
\newtheorem{lemma}{Lemma}
\newtheorem{corollary}{Corollary}

\theoremstyle{definition}
\newtheorem{definition}{Definition}

\newtheorem{remark}{Remark}

\newtheorem{assumption}{}

\newcommand{\CIP}{\stackrel{p}{\longrightarrow}}
\newcommand{\CID}{\stackrel{d}{\longrightarrow}}
\newcommand{\nri}{n\rightarrow \infty}

\title{Synthetic Data Generation With Incomplete Survey Data Under Informative Sampling}
\author{
	Ayat Almomani$^{1}$, Won Chang$^{2}$, Youngdeok Hwang$^{3}$, \\
	Young Min Kim$^{4}$, and Hang J. Kim$^{*5}$\thanks{Corresponding author. E-mail: hang.kim@uc.edu} \\
	\vspace{0.5em}
	\small $^{1}$Yarmouk University, $^{2}$Seoul National University, \\
	\small $^{3}$The City University of New York, $^{4}$Kyungpook National University, \\
	\small $^{*5}$University of Cincinnati
}
\date{}

\begin{document}
	
	\maketitle
	
	\begin{abstract}
		We propose a Bayesian framework for data synthesis and imputation in complex survey settings with informative sampling. To address variance underestimation in existing Bayesian approaches and to accommodate the missing data encountered in survey data, we introduce an adaptive weighting scheme for parameter estimation. We show that the proposed weighting yields consistent estimators with an asymptotically valid Godambe information matrix. The framework is flexible, accommodating a broad class of Bayesian models and facilitating practical implementation. Simulation studies demonstrate that the proposed method provides accurate uncertainty quantification for both model parameters and synthetic population inference. 
		
		\vspace{1em}
	\end{abstract}

	\section{Introduction} \label{sec:introduction}
	
	Since their introduction, synthetic data have been widely recognized as a powerful tool for sharing sensitive data held by governments or industry with the public. For a comprehensive overview of the development and applications of synthetic data, we refer readers to recent review papers by \cite{reiter2023} and \cite{drechsler2024}. To accommodate complex data structures and imperfect data quality with missing values and measurement errors, a variety of synthesis strategies have been developed. For example, tree-based methods have proven useful due to their flexibility in generating synthetic data that preserves key features of the original data, including nonlinear relationships, interactions, and mixed data types \citep{reiter2005, Drechsler2011, NowokRaab2016}. Bayesian nonparametric models have also been employed to synthesize continuous data \citep{kim2018JAS}, mixed-type data \citep{Murray2016,lee2026synmicrodata}, and household-level data \citep{Hu2018BA}, offering a coherent framework to evaluate uncertainty from both sampling variability and the synthesis process. More recently, generative modeling approaches, such as generative adversarial networks and diffusion-based models, have been explored for generating synthetic tabular data, potentially offering enhanced scalability and the ability to capture complex multivariate dependencies \citep{Zilong2021, Kang2023, kotelnikov2023tabddpm}. 
	
	Despite this active line of research, the synthesis of survey data collected under complex sampling designs remains relatively underexplored. An early contribution by \cite{dong2014nonparametric} proposed a nonparametric framework for generating synthetic finite populations under complex survey design. Later, \cite{zhou2016multiple} developed the weighted finite population Bayesian bootstrap for imputation and inference in two-stage cluster samples. Recognizing the distinction between the finite-population and the sampling distributions, \cite{savitsky2016bayesian} introduced a pseudo-posterior framework that incorporates sampling weights via a weighted likelihood formulation. \cite{kim2021synthetic} adopted this framework to generate synthetic finite populations of establishment survey data. More recently, \cite{oganian2024} compared the pseudo-likelihood-based data synthesis with a strategy that incorporates design variables directly into data synthesis models. 
	
	Although the pseudo-posterior-based approach offers a principled way to incorporate sampling weights into Bayesian inference, naive implementation of the pseudo-likelihood can lead to incorrect variance estimation. The issue is due to the discrepancy between the weighted likelihood and the actual data-generating process of the finite population. To mitigate this problem, \cite{leon2019} proposed jointly modeling the survey outcome and the inclusion probabilities as an alternative to the pseudo-posterior approach. In contrast, \cite{williams2021bayesian} developed a post-processing variance correction using the Markov chain Monte Carlo (MCMC) draws from the pseudo-posterior distribution. 
	
	In many official surveys, however, the released microdata are collected under complex sampling designs, often involving missing values or measurement errors. In such settings, Bayesian synthetic data generation requires joint parameter estimation and imputation within the MCMC procedure. Existing variance-adjustment approaches typically rely on post hoc corrections to pseudo-posterior draws.  Thus, these approaches are incompatible with the imputation step in a unified Bayesian synthesis framework. In this article, we propose a pseudo-posterior synthetic-data generation method for incomplete survey samples collected under informative sampling designs. We showed that a simple magnitude adjustment \citep{ribatet2012bayesian} yields asymptotically correct variance estimation and can be implemented directly within flexible Bayesian models, including, for example, Dirichlet process Gaussian mixture models. 
	
	The remainder of the article is organized as follows. Section \ref{sec:background} gives a brief overview of the background. Section \ref{sec:methodology} presents the proposed methodology and derives its theoretical properties. 
	In Section \ref{sec:simulation}, simulation results demonstrate that the method provides accurate uncertainty quantification for the model parameter estimation and with the synthetic populations. 
	We conclude with a summary and discussion in Section \ref{sec:discussion}.

	\section{Background}  \label{sec:background}
	
	\subsection{Synthetic data generation with pseudo-posterior distributions under informative sampling}
	
	Let $\bmY_N = (\bmy_1, \ldots, \bmy_N)$ denote a finite population of $N$ identifiable units whose values are generated from a superpopulation model $f(\cdot|\bmtheta_0)$. Under a survey design, a sample $\bmY_n = \left\{\bmy_i: i \in \mathcal{S} \right\}$ is observed, where $\mathcal{S}$ denotes the index set of sampled units and $\bmw_n = \left\{w_i: i \in \mathcal{S} \right\}$ are the corresponding sampling weights. The observed survey data are assumed to be confidential and therefore cannot be released in their original form.

Suppose an analyst posits a parametric or Bayesian nonparametric model $\{ f(\cdot|\bmtheta): \bmtheta \in \bmTheta \}$ to characterize the distribution of $\bmY_N$, as well as the underlying data-generating process. Under a Bayesian joint modeling framework, a synthetic record $\tilde{\bmy}$ is generated from the posterior predictive distribution
\begin{equation} \label{eqn:predictive-distribution}
	f(\tilde{\bmy} | \bmY_n) = \int f(\tilde{\bmy} | \bmtheta) \pi(\bmtheta | \bmY_n ) \textrm{d} \bmtheta    
\end{equation}
where $\pi(\bmtheta | \bmY_n )$ denotes the posterior distribution of $\bmtheta$ given the observed sample. 

When observations $\bmy_i \in \bmY_n$ are independent and identically distributed or arise from simple random sampling, the posterior distribution takes the familiar form
\begin{equation}
	\label{eq:posterior_iid}
	\pi(\bmtheta | \bmY_n ) \propto \prod_{i=1}^n f( \bmy_i | \bmtheta ) \pi(\bmtheta) 
\end{equation}
where $\pi(\bmtheta)$ is the prior distribution of $\bmtheta$. Synthetic data generated from \eqref{eqn:predictive-distribution} using \eqref{eq:posterior_iid} preserves key distributional features of the finite population $Y_N$.  

However, many survey datasets arise from complex, informative sampling designs in which inclusion probabilities depend on variables related to the outcome of interest. Under informative sampling, the sample distribution generally differs from the finite population (or superpopulation) distribution. Applying the likelihood in \eqref{eq:posterior_iid} without accounting for the sampling mechanism may therefore yield biased parameter estimates and mis-calibrated uncertainty, leading synthetic data to misrepresent population-level relationships. Sampling weights contain essential information for reconciling discrepancies between the observed sample and the target population; ignoring them can distort distributions, dependence structures, and variance estimates. In Figure \ref{fig:density}, the density $f( \bmy_i | \bmtheta)$ under simple random sampling resembles the population distribution. However, under Poisson sampling, it deviates from the population due to its informative sampling design, which is evidenced by the disappearance of the local mode in the bottom-left corner.

\begin{figure}[t!]
	\includegraphics [scale=0.6]{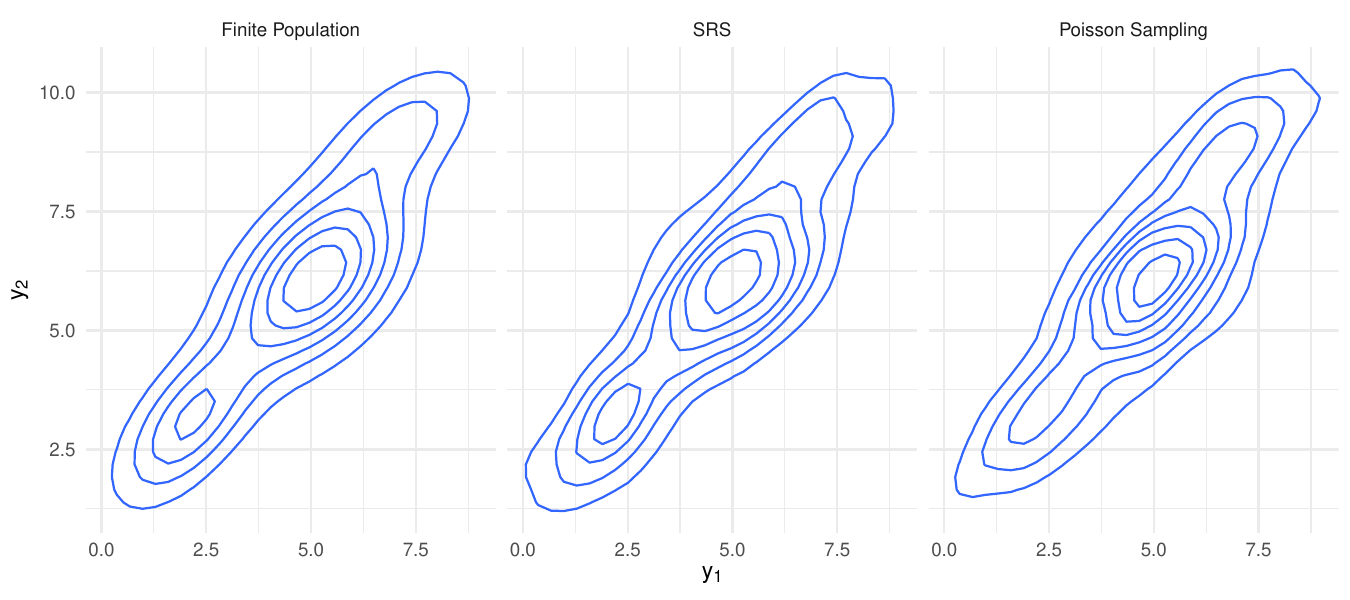}\par
	\caption{Densities of a finite population (left), sampling distribution of simple random sampling (center), and that under Poisson sampling (right).}
	\label{fig:density}
\end{figure}

To incorporate sampling weights within a Bayesian framework, we adopt the pseudo-likelihood approach. \cite{savitsky2016bayesian} proposed constructing a pseudo-posterior by exponentiating each unit’s likelihood function with its corresponding sampling weight. This framework has been extended to unit-level models for non-Gaussian data \citep{parker2020conjugate} and utilized to generate synthetic populations for establishment survey data \citep{kim2021synthetic}.

Formally, following \cite{savitsky2016bayesian}, the pseudo-posterior distribution with the sampling weights is defined as 
\begin{equation} \label{eqn:pseudo-posterior}
	\pi(\bmtheta|\bmY_n, \bmw_n) \propto \prod_{i \in \mathcal{S}} f(\bmy_i|\bmtheta)^{w_i} \cdot \pi(\bmtheta).
\end{equation}
This formulation assigns larger weights to observations representing larger segments of the population and yields approximately unbiased inference under unequal-probability sampling. Sampling weights in the exponents play an analogous role in design-based inference, where they are used to construct design-unbiased estimators. However, the pseudo-posterior in \eqref{eqn:pseudo-posterior} treats each observation $w_i$ as if it were replicated $w_i$ times in the likelihood. This implicit replication inflates the sample information, producing a pseudo-posterior that is typically under-dispersed relative to the true posterior. Consequently, posterior uncertainty can be underestimated, leading to overly narrow credible intervals and overstated inferential precision.

\subsection{Variance estimation in pseudo-Bayesian inference}

Despite the philosophical differences between frequentist and Bayesian inference, their large-sample behaviors are closely related. Under standard regularity conditions, the Bernstein–von Mises theorem implies that the posterior distribution is asymptotically normal, centered at the maximum likelihood estimator, with a covariance equal to the inverse of the Fisher information. Consequently, posterior mean and standard deviation behave similarly to the maximum likelihood estimator and its associated standard error in large samples, under the (frequentist) repeated sampling framework. 

When the likelihood function is misspecified or replaced by a working likelihood, such as the composite likelihood or the pseudo-likelihood, this asymptotic equivalence no longer guarantees valid uncertainty quantification. In such settings, frequentist sampling variance estimation of Bayesian estimates generally requires a sandwich-type variance correction, analogous to that used in frequentist M-estimation \citep[e.g.,][]{ribatet2012bayesian,li2024bayesian}.

Uncertainty quantification under misspecified likelihoods in Bayesian inference has been extensively studied. Early work on a {\it coverage proper} Bayesian likelihood was discussed by \cite{Monahan1992}. Subsequently, \cite{chernozhukov2003mcmc} established asymptotic normality results for {\it quasi-posterior} distributions based on nonlikelihood-based objective functions. For Bayesian inference from composite likelihoods, \cite{ribatet2012bayesian} proposed two posterior-variance adjustment strategies, curvature adjustment and magnitude adjustment, based on asymptotic theory for composite likelihood ratio tests.
A particularly convenient approach is the open-faced sandwich adjustment introduced by \cite{shaby2014open}. Let $\bmtheta^{(1)},\ldots, \bmtheta^{(J)}$ denote posterior samples obtained from Markov chain Monte Carlo (MCMC) samples targeting a quasi-posterior distribution. The calibrated samples are constructed via the affine transformation $\bmtheta_{adj}^{(j)} = \hat{\bmtheta}_{\pi} + \mathbf{C}_n \left(  {\bmtheta}^{(j)} -  \hat{\bmtheta}_{\pi}   \right)$ where $\hat{\bmtheta}_{\pi}$ is the quasi-Bayes estimator minimizing the quasi-posterior risk and $\mathbf{C}_n$ is an adjustment matrix derived from the Godambe information matrix \citep{godambe1986parameters, wooldridge1999asymptotic, wooldridge2001asymptotic}. Posterior inference is then based on the transformed samples $\bmtheta_{adj}^{(j)}$'s, rather than the original MCMC samples. \cite{chang2015composite} used the Godambe information matrix to adjust the variance estimator in the context of Bayesian computer model calibration using quasi-likelihood. 

In the same vein, pseudo-posterior inference in the survey sampling context also introduces bias in the variance estimation. To address this, \cite{williams2021bayesian} proposed a simple post-processing step similar to the transformation approach of \cite{ribatet2012bayesian}. 
While the post-processing approaches in \cite{shaby2014open} or \cite{williams2021bayesian} provide valid uncertainty quantification for complete data, they have two major limitations. First, they are limited in handling missing data, a highly common issue in survey sampling. Post-processing methods like the open-faced sandwich adjustment require obtaining the pseudo-posterior mean (or mode) $\hat{\bmtheta}_{\pi}$ \emph{before} applying the adjustment. This creates an inherently two-stage procedure: (1) obtain $\hat{\bmtheta}_{\pi}$ by running MCMC with the pseudo-posterior, and (2) apply the variance adjustment. When missing values are present, however, the initial pseudo-posterior inference is based only on the non-missing values and can thus be biased. While one could impute missing values prior to the first stage, such an approach fails to properly propagate the imputation uncertainty through to the final estimates. Second, the derivation of the Godambe information matrix becomes prohibitively complex for even moderately sophisticated models, including those involving latent variables such as Dirichlet process Gaussian mixture models.

\section{Weight-Calibrated Pseudo-Bayesian Inference} 
\label{sec:methodology}

In this section, we investigate Bayesian inference derived under general weighting schemes for correct uncertainty quantification. For the survey variables of the sampled units $\bmY_n$, we define the weight-calibrated pseudo-likelihood function as 
\begin{equation} 
	\label{eqn:weight-aug-pseudo-likelihood}
	f(\bmY_n|\bmxi_n,\bmtheta) = \prod_{i \in \mathcal{S}} f \left( \bmy_i| \bmtheta \right)^{\xi_i}, 
\end{equation}
where $\xi_i$ is a real-valued weight associated with the $i$th observation, and $\bmxi_n = \left\{\xi_i: i \in \mathcal{S} \right\}$. 
Then, the corresponding weight-calibrated posterior distribution is given by   
\begin{equation} \label{eqn:calibrated-pseudo-posterior}
	\pi(\bmtheta|\bmY_n, \bmxi_n) \propto \prod_{i \in \mathcal{S}} f(\bmy_i|\bmtheta)^{\xi_i} \cdot \pi(\bmtheta). 
\end{equation}

To facilitate the theoretical analysis of the inference based on \eqref{eqn:calibrated-pseudo-posterior}, we first introduce the following necessary notation and definitions. Let $L_n(\bmtheta)=\log f(\bmY_n|\bmxi_n,\bmtheta)=\sum_{i\in \mathcal{S}} \xi_i\log f(\bmy_i|\bmtheta)$, and define the corresponding score function $\bms_n(\bmtheta) =\nabla_{\bmtheta}L_n(\bmtheta)$, the Fisher information $\bmJ_n(\bmtheta) = \E\left\{ \bms_n(\bmtheta)\bms_n(\bmtheta)^{\top} \right\}$, and the negative of the expected Hessian matrix $\bmH_n(\bmtheta) = - \E \left\{ \nabla_{\bmtheta\bmtheta^{\top}} L_n(\bmtheta) \right\}$.

\subsection{Technical assumptions} 

To investigate the asymptotic properties derived from the weight-calibrated pseudo-posterior in \eqref{eqn:calibrated-pseudo-posterior}, some assumptions are imposed first. 

\begin{assumption} \label{assumption:parameter} 
	The true parameter $\bmtheta_0$ belongs to the interior of a compact set $\Theta \in \mathbb{R}^q$. 
\end{assumption} 

\begin{assumption} \label{assumption:identifiability} 
	For any $\delta>0$, there exists $\varepsilon>0$ such that 
	
	$ \liminf_{n \to \infty} P \left[ \sup_{||\bmtheta-\bmtheta_0|| \ge \delta} \frac{1}{n} \left\{  L_n(\bmtheta)- L_n(\bmtheta_0) \right\} \le - \varepsilon \right] = 1.$
\end{assumption} 

\begin{assumption} \label{assumption:expansion}
	For $\bmtheta$ in an open neighborhood of $\bmtheta_0$, 
	\begin{enumerate}
		\item[(i)] 
		$L_n(\bmtheta) - L_n(\bmtheta_0) = (\bmtheta-\bmtheta_0)^{\top} \bms_n(\bmtheta_0) - \frac{1}{2}(\bmtheta-\bmtheta_0)^{\top}\left\{ n \bmH_n(\bmtheta_0) \right\} (\bmtheta-\bmtheta_0)+R_n(\bmtheta)$ where
		\begin{eqnarray*}
			&(a)& \limsup_{n\rightarrow \infty} P^{*} \left( \sup_{M/\sqrt{n} \leq ||\bmtheta-\bmtheta_0||\leq \delta} \frac{|R_n(\bmtheta)|}{n||\bmtheta-\bmtheta_0||^2} >\varepsilon \right) <\varepsilon, \\
			&(b)& \limsup_{n \rightarrow \infty} P^{*} \left( \sup_{||\bmtheta-\bmtheta_0||\leq M/\sqrt{n}} |R_n(\bmtheta)|>\varepsilon \right)=0,
		\end{eqnarray*}
		for each $\varepsilon>0$, a sufficiently small $\delta>0$, and large $M>0$,
		
		\item[(ii)] 
		$\bmJ_n(\bmtheta_0)^{-1/2}\bms_n(\bmtheta_0)/\sqrt{n} \stackrel{d}{\longrightarrow} N(\bmzero,\bmI_q)$ for an $q\times q$ identity matrix $\bmI_q$,  and 
		\item[(iii)] 
		$\bmJ_n(\bmtheta_0)=O(1)$ and $\bmH_n(\bmtheta_0)=O(1)$ are uniformly in $n$ positive-definite constant matrices. 
	\end{enumerate}
\end{assumption} 

\begin{assumption} \label{assumption:4} 
	A prior function $\pi:\Theta \rightarrow \mathbb{R}$ is two-times continuously differentiable in a neighborhood of $\bmtheta_0 \in \Theta \subset \mathbb{R}^q$ and satisfies 
	$
	\pi(\bmtheta)=\pi(\bmtheta_0)+ \nabla_{\bmtheta}\pi(\bmtheta_0)^{\top}(\bmtheta-\bmtheta_0)+ R_1(\bmtheta),
	$
	where $R_1(\bmtheta)=O\left(||\bmtheta-\bmtheta_0||^2\right)$ is the remainder term.
\end{assumption}

\subsection{Theorems for weight-calibrated posterior inference} 

\label{sec:theorems}

\noindent Define the local parameter $\bmh$ as $\bmh=\sqrt{n} (\bmtheta-\bmtheta_0) - \bmH_n(\bmtheta_0)^{-1}\bms_n(\bmtheta_0) / \sqrt{n}$, which represents the normalized deviation from $\bmtheta_0$ centered by the normalized random score function. Applying the change of variables formula $\bmtheta= \bmh / \sqrt{n} +\bmtheta_0 + \bmH_n(\bmtheta_0)^{-1}\bms_n(\bmtheta_0) / n$, the localized weighted-calibrated posterior density for $\bmh$ is defined by $\pi^*\left(\bmh |\bmY_n,\bmxi_n\right)$ $ = \pi \left( \left. \bmh/\sqrt{n} + \bmtheta_0 + \bmH_n(\bmtheta_0)^{-1}\bms_n(\bmtheta_0) / n \right. \right. $ $\left| \ \bmY_n,\bmxi_n \right) / \sqrt{n}$. 
Here, the convergence of the posterior distribution is examined in the total variation of moments norm, which is defined as
$
||f||_{\text{TVM}}\equiv \int_{\mathcal{H}}||\bmh||^2 |f(\bmh)|d\bmh, 
$
for a real-valued measurable function $\bmh$ on $\mathcal{H}$.
\begin{theorem} \label{thm:Post} 
	Under \ref{assumption:parameter} - \ref{assumption:4}, as $n \rightarrow \infty$, 
	$
	\|\pi^*(\bmh|\bmY_n,\bmxi_n) - \pi^*(\bmh|\bmY_{\infty},\bmxi_{\infty})\|_{\text{TVM}} \overset{p}{\rightarrow} 0,
	$ 
	where 
	$
	\pi^*(\bmh | \bmY_\infty,\bmxi_\infty) \sim N \left( \bmzero, \bmH_n^{-1}(\bmtheta_0)  \right).
	$
\end{theorem}

Theorem \ref{thm:Post} establishes total variation convergence for the weight-calibrated pseudo-posterior, yielding a Bernstein-von Mises-type result. This asymptotic normality justifies the use of Gaussian approximations for posterior summaries and credible intervals in large-samples. Proofs of the main theorems, along with supporting lemmas, are provided in S1 and S2 of the Online Supplement.

\begin{theorem} \label{thm:LTE} 
	Suppose that $\hat{\bmtheta}$ is a posterior mean of $\bmtheta$ with the posterior $\pi^*(\bmh|\bmY_n,\bmxi_n)$. 
	Under \ref{assumption:parameter} - \ref{assumption:4}, 
	$\sqrt{n} (\hat{\bmtheta}-\bmtheta_0)= \bmU_n + o_p(1)$ and  $\bmJ_n^{-1/2}(\bmtheta_0) \bmH_n(\bmtheta_0) \bmU_n$ $\CID$ $N(0,\bmI_q)$,
	where $\bmU_n= \bmH_n(\bmtheta_0)^{-1}\bms_n(\bmtheta_0) / \sqrt{n}$. Hence,  
	$\bmJ_n^{-1/2}(\bmtheta_0)\bmH_n(\bmtheta_0)$ $\sqrt{n} \left( \hat{\bmtheta}-\bmtheta_0 \right)$ $\CID$  $N(0,\bmI_q).$   
\end{theorem}
As a consequence of Theorem \ref{thm:Post}, Theorem \ref{thm:LTE} establishes the $\sqrt{n}$-consistency and asymptotic normality of an Laplace type estimator (LTE) proposed by \citep{chernozhukov2003mcmc} using the weighted-calibrated posterior. To contextualize this theoretical result, recall that the conventional extremum estimator, defined as $\hat{\bmtheta}_{ex} = \arg\sup_{\bmtheta \in \Theta} f(\bmY_n|\bmxi_n,\bmtheta)$ from \eqref{eqn:weight-aug-pseudo-likelihood}, exhibits a normalized estimation error $\sqrt{n}(\hat{\bmtheta}_{ex}-\bmtheta_0)$ that is first-order equivalent to the linear representation $\bmU_n\equiv \bmH_n(\bmtheta_0)^{-1}\bms_n(\bmtheta_0) / \sqrt{n}$. Crucially, Theorem \ref{thm:LTE} demonstrates that the LTE, despite being derived via integration over the weighted-calibrated posterior rather than direct optimization, shares the same first-order asymptotic expansion with the frequentist extremum estimator $\hat{\bmtheta}_{ex}$. Consequently, the LTE using the weighted-calibrated posterior is asymptotically equivalent to the frequentist extremum estimator, inheriting its asymptotic distribution and efficiency while effectively circumventing the computational hurdles associated with optimizing complex objective functions.

Theorem \ref{thm:LTE} implies that a posterior mean $\hat{\bmtheta}$ is asymptotically normal with sandwich covariance $\bmH_n(\bmtheta_0)^{-1}\bmJ_n(\bmtheta_0)\bmH_n(\bmtheta_0)^{-1} / n$. If the likelihood is correctly specified, then $\bmH_n(\bmtheta_0)=\bmJ_n(\bmtheta_0)$, and this sandwich covariance collapses to the inverse Fisher information $\bmJ_n(\bmtheta_0)^{-1}$. Furthermore, Theorem \ref{thm:Post} establishes a Gaussian shape for the (re-centered) posterior density in the $\bmh$-scale. Theorem ~\ref{thm:LTE} then identifies the center (posterior mean) and its sampling distribution.
Taken together, Theorems \ref{thm:Post} and \ref{thm:LTE} justify the construction of Wald-type credible and confidence sets derived from the weighted-calibrated pseudo-posterior.

\begin{theorem} \label{thm:LSI2}
	Suppose that \ref{assumption:parameter} - \ref{assumption:4} hold. 
	Let  $\hat{\bmtheta} = \int_{\Theta} \bmtheta \ \pi(\bmtheta|\bmY_n,\bmxi_n)d\bmtheta$ denote the posterior mean of $\pi(\bmtheta|\bmY_n,\bmxi_n)$. Define the posterior covariance proxy as 
	$\bmH_n^{-1}(\hat{\bmtheta}) \equiv \int_{\Theta} n(\bmtheta-\hat{\bmtheta})(\bmtheta-\hat{\bmtheta})^{\top} \pi_n(\bmtheta|\bmY_n,\bmxi_n)d\bmtheta,$
	and let $\bmJ_n(\hat{\bmtheta})$ be a consistent estimator of the score variance $\bmJ_n(\bmtheta_0)$. For a continuously differentiable scalar functional $g(\cdot)$, define the one-sided asymptotic confidence interval bound as 
	\begin{equation} 
		\label{eq:ConfidenceInterval}
		c_{g,n}(\alpha) \equiv g(\hat{\bmtheta}) + q_\alpha \sqrt{\frac{1}{n} \left\{ \nabla_{\bmtheta} g(\hat{\bmtheta}) \right \}^{\top} \bmH_n^{-1}(\hat{\bmtheta}) \bmJ_n(\hat{\bmtheta}) \bmH_n^{-1}(\hat{\bmtheta}) \left\{ \nabla_{\bmtheta} g(\hat{\bmtheta}) \right\} },
	\end{equation} 
	where $q_{\alpha}$ is the $\alpha$-quantile of $N(0,1)$.
	Then as $\nri$,
	(a) $\bmH_n(\hat{\bmtheta}) \bmH_n(\bmtheta_0)^{-1}$ $\CIP$  $\bmI_q$,
	and (b) $P\left\{ c_{g,n}(\alpha/2) \le g(\bmtheta_0) \le c_{g,n}(1-\alpha/2) \right\} \rightarrow 1-\alpha.$
\end{theorem}

Theorems ~\ref{thm:LTE} and ~\ref{thm:LSI2} further imply that the posterior distribution from \eqref{eqn:calibrated-pseudo-posterior}, with the original sampling weights $w_i$ used in place of $\xi_i$, is asymptotically normal centered at $\bmtheta_0$, making it an asymptotically unbiased estimator. While its variance can be consistently estimated via post-processing based on the Fisher information for complete data \citep{williams2021bayesian}, this approach fails in the presence of missing data due to a circular dependency: consistent estimation of the Fisher information requires imputation, while imputation in turn depends on accurate estimation of the Fisher information. 

\subsection{Calibrated weights for the coverage proper Bayesian inference} 

The conclusion of Section \ref{sec:theorems} is that the asymptotic variance of $g(\hat{\bmtheta})$ depends on both posterior variance $\bmH_n(\hat{\bmtheta})$ and the Fisher information $\bmJ_n(\hat{\bmtheta})$. While the former is readily available from the MCMC samples, using the original sampling weights in \eqref{eqn:pseudo-posterior} solely relies on $\bmH_n(\hat{\bmtheta})$ and ignores $\bmJ_n(\hat{\bmtheta})$, which incorrectly quantifies the uncertainty of $g(\hat{\bmtheta})$.

To address this issue, we propose the following weight that results in coverage-proper Bayesian likelihood \citep{Monahan1992}.
\begin{definition} \label{def:proposed-weight}
Define the calibrated weight for $i$th observation of the weight-calibrated pseudo-likelihood function in \eqref{eqn:weight-aug-pseudo-likelihood} as:
\begin{equation}\label{eqn:cal-weight}
		\xi_{i} =  \frac{\sum_{j \in \mathcal{S}} w_{j}}{\sum_{j \in \mathcal{S}} w_{j}^{2}} w_{i}.
	\end{equation}
\end{definition}
\noindent Then, the following corollary holds for the pseudo-posterior in \eqref{eqn:calibrated-pseudo-posterior} paired with the calibrated weights in \eqref{eqn:cal-weight}.
\begin{corollary} 
	\label{corollary:weight} With the calibrated weight in \eqref{eqn:cal-weight}, the generalized information equality holds, i.e., $\bmH_n(\bmtheta_0) \bmJ_n^{-1}(\bmtheta_0) = \bmI_q$. Then, the asymptotic $(1-\alpha)$ critical bound in \eqref{eq:ConfidenceInterval} is given by 
	$$ c_{g,n}(\alpha) \equiv g(\hat{\bmtheta}) + q_\alpha \sqrt{ \left\{ \nabla_{\bmtheta} g(\hat{\bmtheta}) \right\}^{\top} \bmH_n^{-1}(\hat{\bmtheta}) \left\{ \nabla_{\bmtheta} g(\hat{\bmtheta}) \right\} / n }.$$
\end{corollary}
\noindent Corollary \ref{corollary:weight} implies that under the calibrated weight, the posterior variance $\bmH_n^{-1}(\hat{\bmtheta})$ alone can be used to construct the credible interval of $g(\bmtheta)$, without the explicit estimation of $\bmJ_n^{-1}(\bmtheta_0)$.

In the survey sampling context, the calibrated weight can be viewed in connection with Kish’s effective sample size \citep{kish1965survey}. When dividing the population into $H$ strata and taking $n_h$ elements from each stratum for $h=1,\ldots,H$, Kish's design effect is $\text{Deff} = n\sum \limits _{h=1}^{H}(n_{h}w_{h}^{2}) / \left( \sum \limits _{h=1}^{H}n_{h}w_{h} \right)^{2}$. By putting $n_{h}=1$ for all strata, \cite{kish1965survey} simplified the formula to $\text{Deff}= n \sum_{i \in \mathcal{S}} w_{i}^{2} / \left( \sum_{i \in \mathcal{S}} w_{i} \right)^{2}$ and derived Kish's effective sample size, $n_{\text{eff}}= (\sum_{i \in \mathcal{S}} w_{i})^{2} / \sum_{i \in \mathcal{S}} w_{i}^{2}$ equivalent to $n_{\text{eff}} = n / \text{Deff}$.

\begin{remark}
	If $\sum_{i \in \mathcal{S}} w_i \approx N$, the relation between the calibrated weights and a Kish's effective sample size $n_{\text{eff}}$ (with $n_{h}=1$ for $h=1,\ldots,H$) is expressed by $\xi_i = w_{i} \ (\sum_{i' \in \mathcal{S}} w_{i'} / \sum_{i' \in \mathcal{S}} w_{i'}^{2}) \ \approx \ n_{\text{eff}} ( w_i / N ).$ The effective sample size $n_{\text{eff}}$ represents the size of a simple random sample that would yield the same precision as the weighted sample. It quantifies the relative efficiency (or inefficiency) of the adopted sampling design compared with simple random sampling (SRS). The second term, $w_{i}/N$, reflects the relative contribution of sampled unit $i$ to the inference; under equal selection probabilities, this quantity reduces to  $1/n$. 
\end{remark} 

When the sample is drawn with SRS, the standard posterior inference in \eqref{eq:posterior_iid} without any weight adjustment is sufficient, since the sampling distribution is representative of the finite-population distribution and, consequently, the superpopulation distribution. In contrast, if one directly uses the pseudo-posterior distribution with the original sampling weights $w_i$ as in \eqref{eqn:pseudo-posterior}, the resulting posterior variance underestimates the frequentist sampling variance of the posterior mean. However, the suggested pseudo-posterior distribution with the calibrated weights automatically addresses this issue. 

\begin{remark}
	If the survey samples are drawn via SRS, so that $w_i = N/n$ for all $i$, then $\xi_i = 1$. In this case, the weight-calibrated posterior distribution $\pi(\bmtheta|\bmY_n, \bmxi_n)$ reduces to $\pi(\bmtheta | \bmY_n ) \propto \prod_{i=1}^n f( \bmy_i | \bmtheta ) \pi(\bmtheta)$, 
	i.e., the posterior inference is made with the standard posterior based on the unweighted likelihood. 
\end{remark}

\noindent Therefore, our proposed approach may be viewed as a safeguard that automatically adjusts to the sampling design, reducing to the usual posterior under SRS while providing the appropriate adjustment under informative sampling.

\section{Simulation Study} \label{sec:simulation}

In this section, we assess the finite-sample performance of the proposed methods through repeated simulation studies designed to corroborate our theoretical results and to evaluate the accuracy of finite-population inference using synthetic populations. 

\subsection{Simulation settings}

Our simulation setup is adopted from \cite{kim2021synthetic}, reflecting the skewed and irregular patterns inherent in business establishment surveys under an informative sampling design. We assume a finite population consisting of $N=10^6$ units with $p=4$ variables, generated from a superpopulation following a Gaussian mixture with three mixture components,  $f(\bmy_i) = \sum_{k=1}^3 \eta_{k} \N \left( \bmy_i; \bmmu_k, \bmSigma_k \right)$ where $\eta_k$, $\bmmu_k$, and $\bmSigma_k$ denote the mixture proportion, mean vector, and covariance matrix of the $k$th component, respectively. The left panel of Figure \ref{fig:density} displays the empirical distributions of the simulated finite population. We refer readers to the Online Supplement S3 for the simulation parameters used in this study. 

Given the finite population, we conducted $L=400$ repeated simulations. In each replication, a sample $\bmY_n$ was selected using SRS and Poisson sampling, respectively. For the Poisson sampling, the inclusion probability of unit $i$ was set as $p_i = s_i \big/ \sum_{i' =1}^N s_{i'}$ where the size measure for unit $i$ was constructed as $s_i = | 0.7 y^*_{i1}  +0.4 y^*_{i2} + 0.5 y^*_{i3} |$ where $y^*_{ij} \sim \N \left( \left\{ y_{ij} -\min(y_{1j},\ldots,y_{Nj}) +0.1 \right\}^{1/2}, b_j \right)$ with $b_1 = b_3 = 0.1$ and $b_2 = 0.01$. The corresponding sampling weight for unit $i$ is $w_i = 1 / p_i$. The center and right panels of Figure \ref{fig:density} display the empirical distributions of samples obtained under SRS and Poisson sampling with the size measures defined above, respectively. 

The Poisson sampling design used here is clearly informative, as evidenced by the discrepancy between the finite population and sample distributions. This discrepancy arises because the inclusion probabilities depend on survey variables, making the design non-ignorable except under highly restrictive conditions \citep{pfeffermann1993,sugden1984ignorable}.  

We also introduced artificial missing values under the missing-at-random assumption
to emulate the high missing rate commonly observed in establishment survey data \citep{kim2018JAS,kim2021synthetic}. Specifically, missingness was imposed on $y_{i2}$ and $y_{i3}$ with missing probabilities $1/\left\{2 + \exp(-y_{i1})\right\}$ and $1/\left\{2 + \exp(-y_{i4})\right\}$, respectively. In the resulting incomplete datasets of size $n=800$, an average of 51.6\% of sampled units contained at least one missing value, and 12.6\% had missing values in two survey variables.

\subsection{Synthetic data generation with a Dirichlet process mixture Gaussian model} 
\label{sec:simul_methods}

Using the simulated sample dataset $\bmY_n$, we generated synthetic populations of size $N=10^6$ from the estimated finite-population density, incorporating the sampling weights into the model. The data synthesis relies on the calibrated pseudo-posterior predictive distributions
$f( \tilde{\bmy} | \bmY_n, \bmxi_n ) = \int f( \tilde{\bmy} | \bmtheta ) \pi(\bmtheta|\bmY_n, \bmxi_n) \diff \bmtheta.$

Among many possible specifications for the density function $f(\bmy_i | \bmtheta)$ of observed records (i.e., synthetic records), a nonparametric Bayesian model provides a flexible framework for inference and imputation under informative sampling, paired with our calibrated pseudo-likelihood weights. Because the sample data do not conform to any simple parametric form, such as a normal distribution, we adopt a Dirichlet process (DP) Gaussian mixture model with truncation \citep{Sethuraman:1994, Ishwaran:2001}. 

We now describe how the proposed calibrated weights are incorporated into the DP Gaussian mixture model. Let the model parameters be $\bmmu_{1:K} = \{\bmmu_1,\ldots,\bmmu_k\}$, $\bmSigma_{1:K} = \{\bmSigma_1,\ldots,\bmSigma_k\}$, and $\bmeta = \{ \eta_1,\ldots,\eta_k \}$. The weight-calibrated pseudo-likelihood function in \eqref{eqn:weight-aug-pseudo-likelihood} can then be written as \begin{equation}  \label{eq:Likelihood}
	f(\bmY_n | \bmxi_n, \bmmu_{1:K}, \bmSigma_{1:K}, \bmeta ) = \prod_{i \in \mathcal{S}}  \left\{ \sum_{k=1}^{K} \eta_k \N \left( \bmy_i ; \bmmu_k, \bmSigma_k \right) \right\}^{\xi_i}.   
\end{equation} 
Introducing a latent variable $z_i$ that indicates the mixture component for unit $i$, Equation \eqref{eq:Likelihood} can be equivalently expressed as
$f(\bmY_n| \bmz_n ,\{\bmmu_k\}, \{\bmSigma_k\}) = \prod_{i \in \mathcal{S}} N \left( \bmy_i ; \bmmu_{z_i}, \bmSigma_{z_i} \right)^{\xi_i}$ and 
$f(\bmz_n | \bmeta) = \prod_{i \in \mathcal{S}} \left( \eta_{1}^{I(z_i=1)} ... \eta_{k}^{I(z_i=k)} \right)^{\xi_i}$ 
where $\bmz_n = \{ z_i : i \in \mathcal{S} \}$.  

The Gaussian parameters are assigned conjugate priors: 
$f(\bmmu_k|\bmSigma_k) \sim \mathcal{N}(\bmmu_0, h_0^{-1} \bmSigma_k)$ and $f(\bmSigma_k| \Phi)  \sim \IW(\zeta_0, \boldsymbol{\Phi})$ where $\boldsymbol{\Phi} = \textrm{diag}(\phi_1,...,\phi_p)$ and $f(\phi_j) \sim \text{Gamma}(a_{\phi}, b_{\phi})$ for $j=1,\ldots,p$. Here, IW denotes the inverse Wishart distribution. For the mixture proportions, we use a truncated stick-breaking representation: 
$\eta_{1}=\nu_1$, $\eta_{k}  = \nu_k \prod_{g=1}^{k-1} (1 - \nu_g)$ for $k=2,\ldots,K$, $\nu_k|\alpha \sim \textrm{Beta}(1,\alpha)$ for $k=1,...,K-1$, $\nu_{K}=1$, and $\alpha \sim \text{Gamma}(a_{\alpha}, b_{\alpha})$. Hyperparameters were chosen as follows: $\bmmu_0$ was set to the zero vector since all variables were mean-centered; $\zeta_0 = p+1$ to ensure a proper but weakly informative prior on $\boldsymbol{\Sigma}_k$; and $h_0 = 1$. We also set $a_\phi = b_\phi = a_\alpha = b_\alpha = 0.25$ to assign modest prior mass to moderate variance values. The truncation level was set to $K=10$, which provided sufficient flexibility to model the simulated data. 

Given a simulated sample dataset $\bmY_n$ containing missing values, we generated multiple synthetic populations using a Metropolis-within-Gibbs sampler for the DP Gaussian mixture model described above. For each simulation replication, we run 2,000 iterations of the MCMC algorithm. After removing the first 1,000 iterations as an initial burn-in period, we draw $N$ values of $\tilde{\bmy}$ as a synthetic population in every $t_\text{int}=10$ iterations, resulting in the $M$=100 synthetic populations. 

\begin{enumerate}
	\item For each $k=1,\ldots,K$, randomly draw $\bmSigma_k$ and $\bmmu_k$ from the conditional densities, 
	$f( \bmSigma_k | \cdots ) \sim \IW \left( \zeta_k,  \bmPhi_k  \right)$
	and $f( \bmmu_k | \bmSigma_k, \cdots ) \sim \N \left( \bmmu_k^*, \bmSigma_k^* \right)$
	where $N_k = \sum_{\{i : z_i=k\}} \ \xi_i$, $\bms_k = \sum_{\{i : z_i=k\}} \xi_i \bmy_i$, $\bmT_k = \sum_{\{i : z_i=k\}} \xi_i  ( \bmy_i- \bms_k/N_k ) ( \bmy_i- \bms_k/N_k )^\top$, $\zeta_k = \zeta_0 + N_k$, $\bmPhi_k = \bmPhi + \bmT_k + ( \bms_k/N_k - \bmmu_0 )  ( \bms_k/N_k - \bmmu_0 )^\top / ( 1/N_k + 1/h_0 )$, $\bmmu_k^* = ( \bms_k + h_0 \bmmu_0 ) / ( N_k + h_0 )$ and $\bmSigma_k^* = \bmSigma_k / (N_k+h_0)$.
	
	\item For each $k=1,\ldots,K-1$, randomly draw $\nu_{k}$ from
	$f( \nu_k | \cdots ) \sim \text{Beta} ( 1 + N_k, \alpha + \sum_{m = k+1}^K N_m )$, and set $\nu_K=1$. Then, $\log \eta_1 = \log \nu_1$ and $\log \eta_k = \log \nu_k + \sum_{m=1}^{k-1} \log (1-\nu_m)$ for $k=2,\ldots,K$.
	
	\item For $j=1,\ldots,p$, draw $\phi_j$ from
	$f( \phi_j | \cdots ) \sim \text{Gamma} (  a_\phi + K  \zeta_0/2, \ b_\phi + \sum_{k=1}^K \sigma_{k,j}^{-2}/2 )$
	where $\sigma_{k,j}^{-2}$ is the $j$-th diagonal element of  $\bmSigma_k^{-1}$.
	
	\item Randomly draw $\alpha$ from 
	$f( \alpha | \cdots ) \sim \text{Gamma}( a_\alpha + K - 1, b_\alpha - \log \eta_K )$.
	
	\item For each $i=1,\ldots,n$, randomly draw $z_i$ from $f(z_i|\cdots) \sim \text{Categorical}(\eta_{i1}^*,\ldots,\eta_{iK}^*)$
	where $ \eta_{ik}^* = \eta_{k} \ \N \left(\bmy_i ; \bmmu_k, \bmSigma_k \right) /  \{ \sum_{g=1}^K \eta_{g} \ \N \left(\bmy_i ; \bmmu_g, \bmSigma_g \right) \}$.
	
	\item For record $i$ with missing items, divide $\bmy_i$ into $\bmy_i = (\bmy_{i,\text{m}},\bmy_{i,\text{o}})$ where $\bmy_{i,\text{m}}$ and $\bmy_{i,\text{o}}$ denote items with missing and observed values, respectively. We also divide the corresponding mean vector and covariance matrix as $\bmmu_{z_i}^\top = (\bmmu_{i,\text{m}}^\top, \bmmu_{i,\text{o}}^\top)$ and $\bmSigma_{z_i} = \begin{pmatrix} \bmSigma_{i,\text{mm}} & \bmSigma_{i,\text{mo}} \\
		\bmSigma_{i,\text{mo}}^\top & 
		\bmSigma_{i,\text{oo}}
	\end{pmatrix}.$
	
	Then, randomly draw $\bmy_{i,\text{m}}$ from the conditional normal distribution $f( \bmy_{i,\text{m}} | \bmy_{i,\text{o}}, \cdots ) \sim \N (\bmmu_i^*, \bmSigma_i^* )$ where 
	$\bmmu_i^* = \bmmu_{i,\text{m}} + \bmSigma_{i,\text{mo}} \bmSigma_{i,\text{oo}}^{-1} (\bmy_{i,\text{o}} - \bmmu_{i,\text{o}})$ and 
	$\bmSigma_i^* = \bmSigma_{i,\text{mm}} - \bmSigma_{i,\text{mo}} \bmSigma_{i,\text{oo}}^{-1} \bmSigma_{i,\text{mo}}^\top$. 
	
	\item For every $t_\text{int}=10$ iterations, draw $N=10^6$ synthetic population units $\tilde{\bmy}_i$ with the following steps: 
	\begin{enumerate}
		\item For each $k$, compute $\tilde{N}_k = \text{round}(N \eta_k)$.
		\item Draw $\tilde{N}_k$ synthetic records from their posterior predictive distribution $f(\tilde{\bmy}_i|\cdots) \sim \sum_{k=1}^{K} \eta_k \N \left( \tilde{\bmy} ; \bmmu_k, \bmSigma_k \right)$. 
	\end{enumerate}
	
\end{enumerate}

Repeat Steps 1 through 7 for the total number of iterations. 

\subsection{Simulation results} 

We compared the performance of the proposed method with two alternatives currently used in practice: (1) synthetic data generation without incorporating survey weights, which relies on the i.i.d. likelihood function in $\eqref{eq:posterior_iid}$; and (2) synthetic data generation using the original sampling weights based on the pseudo-likelihood in \eqref{eqn:pseudo-posterior}. 

To evaluate the fit of DP Gaussian mixture models under the different weighting schemes, we examined inference for the superpopulation parameters by estimating $E(Y_j) = \sum_{k=1}^K \eta_k \mu_{k,j}$ for $j \in \{1,2,3,4\}$, and compared the corresponding posterior means and 95\% equal-tailed credible intervals across methods. 
Table \ref{tab:super-pop} summarizes the performance of the competing methods in the simulation study.
The posterior inference without incorporating weights (``no weight'') shows that failing to account for sampling weights leads to substantial bias. Although using the original sampling weights (``sampling weight'') yields approximately unbiased point estimates of the mean, it underestimates the associated posterior variance. As a result, both approaches produce invalid credible interval coverages. In contrast, the proposed method (``calibrated weight'') provides accurate inferences, achieving coverage close to the nominal level.

\begin{table}[t!] 
	\caption{The results for estimating the superpopulation parameters from repeated simulations with synthetic data models, using the posterior inference without using weights, using original sampling weights, and using calibrated weights. The results compare the bias, the mean squared error (MSE), and the 95\% equal-tailed credible interval coverage (Coverage).}
	\label{tab:super-pop}\par
	\begin{center}
		\renewcommand{\arraystretch}{0.8} 
		\begin{tabular}{|rcrcc|} \hline 
			&    & Bias  & MSE  & Coverage \\  
			[3pt] \hline
			& no weight & 0.40 & 0.17 & 0.045 \\
			$E(Y_1)$     & sampling weight & 0.00 & 0.01 & 0.030 \\
			& calibrated weight & 0.01 & 0.02 & 0.925 \\ 
			\hline
			& no weight & 0.25 & 0.08 & 0.095 \\
			$E(Y_2)$     & sampling weight & -0.03 & 0.01 & 0.180 \\ 
			& calibrated weight & 0.00 & 0.01 & 0.950 \\ 
			\hline
			& no weight & 0.32 & 0.12 & 0.110 \\
			$E(Y_3)$     & sampling weight & -0.01 & 0.01 & 0.303 \\
			& calibrated weight & 0.00 & 0.01 & 0.970 \\ 
			\hline
			& no weight & 0.45 & 0.23 & 0.068 \\
			$E(Y_4)$     & sampling weight & 0.01 & 0.02 & 0.043 \\
			& calibrated weight & 0.01 & 0.03 & 0.913 \\ 
			\hline
		\end{tabular}
	\end{center}
\end{table}

Next, we assessed how the accuracy of parameter estimation translates into synthetic data quality by examining finite-population inference based on the synthetic population. For each simulation replication, we followed the procedures described in Section \ref{sec:simul_methods} to generate $M=100$ multiple synthetic populations, denoted by $\tilde{\bmY}_N^{(m)}$ for $m=1,\ldots,M$. The finite population mean of the $j$th variable, $\overline{Y}_{N,j}$ for $j=1,\ldots,4$, was estimated as $\widehat{Y}_{N,j} = \sum_{m=1}^M \widehat{Y}_{N,j}^{(m)} / M$ where $\widehat{Y}_{N,j}^{(m)}$ denotes the population mean of $j$th variable in the $m$th synthetic population. The corresponding 95\% confidence interval was constructed as $\widehat{Y}_{N,j} \pm t_{0.025,M-1} \sqrt{ \widehat{V} ( \widehat{Y}_{N,j} ) }$, where the variance estimator for multiple synthetic populations \citep{raghunathan2003} is given by $\widehat{V} ( \widehat{Y}_{N,j} ) = (1 + 1/M) \sum_{m=1}^M  ( \widehat{Y}_{N,j}^{(m)} - \widehat{Y}_{N,j} )^2 / (M-1)$. Table \ref{tab:finite-inference} summarizes the results of this finite-population inference across the three methods. 

\begin{table}[t!] 
	\caption{The results for estimating the finite population means, using the data synthesis model without using weights, using original sampling weights, and using calibrated weights. The results are compared with the bias, the mean squared error (MSE), and the 95\% confidence interval coverage (Coverage).}
	\label{tab:finite-inference}\par
	\begin{center}
		\renewcommand{\arraystretch}{0.8} 
		\begin{tabular}{|rcrcc|} \hline 
			&    & Bias  & MSE  & Coverage \\  
			[3pt] \hline
			& no weight & 0.41 & 0.25 & 0.113 \\
			$\overline{Y}_{N,1}$    & sampling weight & 0.00 & 0.01 & 0.040 \\
			& calibrated weight & 0.01 & 0.02 & 0.880 \\
			\hline
			& no weight & 0.25 & 0.07 & 0.200 \\ 
			$\overline{Y}_{N,2}$    & sampling weight & -0.03 & 0.01 & 0.183 \\
			& calibrated weight & 0.00 & 0.01 & 0.933 \\ 
			\hline
			& no weight & 0.34 & 0.35 & 0.225 \\ 
			$\overline{Y}_{N,3}$    & sampling weight & -0.01 & 0.01 & 0.298 \\
			& calibrated weight & 0.00 & 0.02 & 0.928 \\  
			\hline
			& no weight & 0.46 & 0.23 & 0.128 \\ 
			$\overline{Y}_{N,4}$    & sampling weight & 0.00 & 0.02 & 0.045 \\
			& calibrated weight & 0.01 & 0.03 & 0.893 \\ 
			\hline
		\end{tabular}
	\end{center}
\end{table}

The results demonstrate that existing methods fail to produce valid synthetic data. Ignoring sampling weights leads to biased mean estimates and inaccurate confidence intervals for the finite population parameters, while using unadjusted weights substantially underestimates variance. Consequently, the resulting synthetic data are generated from a distribution that is overly concentrated around the mean, leading to severe undercoverage. In contrast, the proposed approach with calibrated weights achieves substantially improved coverage.

\section{Concluding Remarks} \label{sec:discussion}

In this article, we proposed a Bayesian model-based inference framework with incomplete survey data under informative sampling. The proposed method is computationally straightforward and sufficiently flexible to be integrated with a wide range of modeling strategies. Both theoretical investigations and simulation studies demonstrate its effectiveness for model-based inference tasks, particularly in applications such as missing data imputation and synthetic data generation. 

We conclude by outlining several directions for future research. First, an important extension is to incorporate the proposed framework into modern generative modeling approaches for synthetic data generation. Although generative models have received considerable attention for their ability to handle complex and high-dimensional data, existing methods typically lack a principled mechanism for incorporating sampling weights. The calibrated weights proposed in this work could be adapted to these generative models to fully leverage the information contained in sampling weights.

Second, while the proposed method substantially improves coverage compared to existing approaches, a modest gap from the nominal level remains. This discrepancy likely arises from several unresolved challenges, including the inherent clustering uncertainty in the Dirichlet process mixture model. A more refined theoretical analysis of this phenomenon, along with methodological developments to further improve calibration, represents an important direction for future work.

\par

\bibhang=1.7pc
\bibsep=2pt
\renewcommand\bibname{\large \bf References}
\expandafter\ifx\csname
natexlab\endcsname\relax\def\natexlab#1{#1}\fi
\expandafter\ifx\csname url\endcsname\relax
\def\url#1{\texttt{#1}}\fi
\expandafter\ifx\csname urlprefix\endcsname\relax\def\urlprefix{URL}\fi

\bibliographystyle{chicago}
\bibliography{paper-ref}

\vskip .65cm
\noindent
Ayat Almomani, Department of Statistics, Yarmouk University, Jordan
\vskip 2pt
\noindent
E-mail: ayat.momani@yu.edu.jo
\vskip 2pt

\noindent
Won Chang, Department of Statistics and the Institute for Data Innovation in Science, Seoul National University, Korea
\vskip 2pt
\noindent
E-mail: wonchang@snu.ac.kr

\noindent
Youngdeok Hwang, Paul H. Chook Department of Information Systems and Statistics, City University of New York, U.S.A.
\vskip 2pt
\noindent
E-mail: youngdeok.hwang@baruch.cuny.edu

\noindent
Young Min Kim, Department of Statistics, Kyungpook National University, Korea
\vskip 2pt
\noindent
E-mail: kymmyself@knu.ac.kr

\noindent
Hang J. Kim, Division of Statistics and Data Science, University of Cincinnati, U.S.A.
\vskip 2pt
\noindent
E-mail: hang.kim@uc.edu

\newpage

\begin{center}
	\setcounter{page}{1} 
	\setcounter{section}{0}
	\setcounter{equation}{0}
	\def\theequation{S.\arabic{equation}}
	\def\thesection{S.\arabic{section}}
	
	{\Large\bf Supplementary Material: \\ Synthetic Data Generation Under Informative Sampling \par}
	\vspace{1.5em}
	{\large Ayat Almomani$^{1}$, Won Chang$^{2}$, Youngdeok Hwang$^{3}$, Young Min Kim$^{4}$, and Hang J. Kim$^{*5}$ \par}
	\vspace{1em}
	{\small $^{1}$Yarmouk University, $^{2}$Seoul National University, $^{3}$The City University of New York, \\ $^{4}$Kyungpook National University, and $^{*5}$University of Cincinnati \par}

\end{center}

\vspace{3em}

	\section*{Proofs of Theorems}
	
	\setcounter{equation}{0}
	
	\subsection*{Proof of Theorem 1}
	
	Define the total variation of moments norm for any real-valued measurable function $f$ on $\mathcal{H}$ is 
	\[
	||f||_{\text{TVM}}\equiv \int_{\mathcal{H}}||\bmh||^2 |f(\bmh)|d\bmh. 
	\]

	Let
	\[
	\mathcal{H}_n = \{ \sqrt{n}(\bmtheta - \bmtheta_0 ) - \bmH_n(\bmtheta_0)^{-1}\bms_n(\bmtheta_0)/\sqrt{n} : \bmtheta \in \Theta \}. 
	\]
	Then,
	\[
	||\pi^*(\bmh|\bmY_n,\bmxi_n) - \pi^*(\bmh|\bmY_{\infty},\bmxi_{\infty})||_{\text{TVM}} =
	\int_{\mathcal{H}_n} ||\bmh||^2 \left| \pi^*(\bmh | \bmY_n,\bmxi_n) - \pi^*(\bmh | \bmY_\infty,\bmxi_\infty) \right| d \bmh.
	\]
	Hence, Lemma~\ref{lm:conv_total_var} implies that 
	\[
	||\pi^*(\bmh|\bmY_n,\bmxi_n) - \pi^*(\bmh|\bmY_{\infty},\bmxi_{\infty})||_{\text{TVM}} \CIP 0 \quad \text{as} \quad \nri.
	\]

	\subsection*{Proof of Theorem 2}
	
	Note that $\bmH_n(\bmtheta)\equiv - \E \left[\nabla_{\bmtheta\bmtheta^T}L_n(\bmtheta) \right]$ and $\bmV(\bmtheta) \equiv \E \left[ 
	\bms_n(\bmtheta) \bms_n(\bmtheta)^T \right]$ which do not depend on the sample size $n$. Recall
	\[
	\bmh = \sqrt{n}(\bmtheta-\bmtheta_0) - \bmH_n(\bmtheta_0)^{-1}\bms_n(\bmtheta_0)/\sqrt{n} \quad \text{and} \quad  \bmU_n=n^{-1/2}\bmH_n(\bmtheta_0)^{-1}\bms_n(\bmtheta_0).
	\]
	Consider the convex objectives
	$Q_n(\bmz) = \int_{\mathcal{H}_n} (\bmz-\bmh-\bmU_n)^2 \pi^{*}(\bmh|\bmY_n,\bmxi_n)d\bmh$ and
	$Q_{\infty}(\bmz) = \int_{\mathbb{R}^q} (\bmz-\bmh-\bmU_n)^2$ 
	$\pi^{*}(\bmh|\bmY_{\infty},\bmxi_{\infty})d\bmh$,
	which are minimized at $\sqrt{n}(\hat{\bmtheta}-\bmtheta_0)$ and a random vector $\bmZ_n$, respectively where $\mathcal{H}_n = \{ \sqrt{n}(\bmtheta - \bmtheta_0 ) - \bmH_n(\bmtheta_0)^{-1}\bms_n(\bmtheta_0)/\sqrt{n} : \bmtheta \in \Theta \}$. By Lemma~\ref{lm06} (Anserson's inequality),
	\[
	\arg\inf_{\bmz \in \mathbb{R}^q} \left\{ \int_{\mathbb{R}^q} (\bmz-\bmh)^2 \pi^{*}(\bmh|\bmY_{\infty},\bmxi_{\infty})d\bmh  \right\}=0.
	\]
	Hence, we have 
	\begin{equation}\label{eq:zn}
		\bmZ_n =\bmU_n+O_p(1). 
	\end{equation} 
	
	To show 
	\[
	\sup_{\bmz \in \mathbb{R}^q}\left| Q_n(\bmz)-Q_{\infty}(\bmz) \right| \CIP 0  \quad \text{as} \quad \nri,
	\]
	we should first prove, for ant fixed $\bmz$, 
	\[
	\left| Q_n(\bmz)-Q_{\infty}(\bmz) \right| \CIP 0 \quad \text{as} \quad \nri. 
	\]
	Thus, we decompose it as 
	\begin{eqnarray*}
		\left| Q_n(\bmz)-Q_{\infty}(\bmz) \right| & \leq  & \int_{\mathcal{H}_n} (\bmz-\bmh-\bmU_n)^2 \left|\pi^{*}(\bmh|\bmY_n,\bmxi_n)-\pi^{*}(\bmh|\bmY_{\infty},\bmxi_{\infty}) \right|d\bmh \\
		&  & \quad \quad + \int_{\mathcal{H}_n^c} (\bmz-\bmh-\bmU_n)^2 \pi^{*}(\bmh|\bmY_{\infty},\bmxi_{\infty})d\bmh \\
		&\equiv& B_{1n}+B_{2n}
	\end{eqnarray*}
	where $\mathcal{H}_n^c$ is the complement of $\mathcal{H}_n$. Using $(a+b)^2 \leq 2a^2+2b^2$, we have 
	\begin{eqnarray*}
		B_{1n}&\leq& \int_{\mathcal{H}_n} \left\{2(\bmz-\bmU_n)^2 +2\bmh^2  \right\} \left|\pi^{*}(\bmh|\bmY_n,\bmxi_n)-\pi^{*}(\bmh|\bmY_{\infty},\bmxi_{\infty}) \right|d\bmh \\
		& =& 2(\bmz-\bmU_n)^2 \int_{\mathcal{H}_n} \left|\pi^{*}(\bmh|\bmY_n,\bmxi_n)-\pi^{*}(\bmh|\bmY_{\infty},\bmxi_{\infty}) \right|d\bmh \\
		& & +2\int_{\mathcal{H}_n} \bmh^2  \left|\pi^{*}(\bmh|\bmY_n,\bmxi_n)-\pi^{*}(\bmh|\bmY_{\infty},\bmxi_{\infty}) \right|d\bmh\\
		&= &  o_p(1) + o_p(1) = o_p(1)
	\end{eqnarray*}
	by the proof of Lemma~\ref{lm:conv_total_var} with $\alpha=0$ and $\alpha=2$ using Lemma~\ref{lm03}. Since the Lebesgue measure of $\mathcal{H}_n^c$ converges to zero, $B_{2n}=o_p(1)$ as $\nri$. Hence, as $\nri$,
	\[
	\left| Q_n(\bmz)-Q_{\infty}(\bmz) \right| \CIP 0. 
	\]
	Since $Q_n(\bmz)$ and $Q_{\infty}(\bmz)$ are convex and finite, by equation \eqref{eq:zn}, we can re-express
	\[
	\bmZ_n =\arg\int_{\bmz\in\mathbb{R}^q} Q_{\infty}(\bmz) =O_p(1). 
	\]
	By Lemma~\ref{lm07}, the pointwise convergence entails the uniform convergence over compact sets $\mathbb{K}$: 
	\[
	\sup_{\bmz \in \mathbb{K}} \left| Q_n(\bmz)-Q_{\infty}(\bmz) \right| \CIP 0 \quad \text{as} \quad \nri. 
	\]
	Since $\bmZ_n=O_p(1)$, uniform convergence and convexity arguments imply that $\sqrt{n}(\hat{\bmtheta}-\bmtheta_0) - \bmZ_n \CIP 0$ by Lemma~\ref{lm05}. 
	
	The proof is complete.

	\subsection*{Proof of Theorem 3}
	
	Recall that $\bmh = \sqrt{n}\left(\bmtheta -\bmtheta_0  \right) - \bmU_n$ and $\bmU_n = \frac{\bmH_n(\bmtheta_0)^{-1}\bms_n(\bmtheta_0)}{\sqrt{n}}$, so that 
	\[
	\bmtheta= \bmtheta_0 + \frac{\bmh+\bmU_n}{\sqrt{n}}.
	\]
	
	A change of variables gives
	\begin{equation}\label{eq:Jhat}
		\bmH_n^{-1}(\hat{\bmtheta}) = \int_{\mathcal{H}_n} (\bmh-\bmT_n)(\bmh-\bmT_n)^T\pi^{*}\left(\bmh|\bmY_{n},\bmxi_{n} \right) d\bmh, \quad \text{and} \quad \bmT_n \equiv \sqrt{n}(\hat{\bmtheta}_n-\bmtheta_0) - \bmU_n,
	\end{equation}
	where $\hat{\bmtheta}_n$ is a posterior mean. Likewise, the limiting curvature inverse has the representation as 
	\begin{equation}\label{eq:Jn}
		\bmH_n(\bmtheta_0)^{-1} = \int \bmh \bmh^T \pi^{*}\left(\bmh|\bmY_{\infty},\bmxi_{\infty} \right) d\bmh. 
	\end{equation}
	By Equation~\eqref{eq:Un}, we have $\sqrt{n}(\hat{\bmtheta}_n-\bmtheta_0)=\bmU_n+o_p(1)$. Hence, we also have $\bmT_n=o_p(1)$.

	In order to show $\bmH_n(\hat{\bmtheta})^{-1} - \bmH_n(\bmtheta_0)^{-1} \CIP 0$ as $\nri$, we first fix a truncation set $\mathcal{H}_n=\left\{||\bmh||\leq n^{1/6}  \right\}$. By using \eqref{eq:Jhat}, and \eqref{eq:Jn},  we decompose $|\bmH_n(\hat{\bmtheta})^{-1} - \bmH_n(\bmtheta_0)^{-1} |$ into the following six integrals:
	\begin{eqnarray*}
		&  & \Bigg| \int_{\mathcal{H}_n} \left(\bmh\bmh^T-2\bmT_n\bmh^T+\bmT_n\bmT_n^T\right)\pi^{*}\left(\bmh|\bmY_{n},\bmxi_{n} \right) d\bmh - \int \bmh \bmh^T \pi^{*}\left(\bmh|\bmY_{\infty},\bmxi_{\infty} \right) d\bmh    \Bigg|  \\
		& \leq &  \Bigg| \int_{\mathcal{H}_n}  \bmh\bmh^T \left\{\pi^{*}\left(\bmh|\bmY_{n},\bmxi_{n} \right)-  \pi^{*}\left(\bmh|\bmY_{\infty},\bmxi_{\infty} \right) \right\}    d\bmh \Bigg|+\Bigg| \int_{\mathcal{H}_n^c} \bmh \bmh^T \pi^{*}\left(\bmh|\bmY_{\infty},\bmxi_{\infty} \right) d\bmh  \Bigg| \\
		&      & + \Bigg| \int_{\mathcal{H}_n}  \bmT_n\bmT_n^T \left\{\pi^{*}\left(\bmh|\bmY_{n},\bmxi_{n} \right)-  \pi^{*}\left(\bmh|\bmY_{\infty},\bmxi_{\infty} \right) \right\}        d\bmh   \Bigg| + \Bigg|  \int_{\mathcal{H}_n}  \bmT_n\bmT_n^T \pi^{*}\left(\bmh|\bmY_{\infty},\bmxi_{\infty} \right) d\bmh  \Bigg| \\
		&     &  + \Bigg| \int_{\mathcal{H}_n}  \bmh\bmT_n^T \left\{\pi^{*}\left(\bmh|\bmY_{n},\bmxi_{n} \right)-  \pi^{*}\left(\bmh|\bmY_{\infty},\bmxi_{\infty} \right) \right\}        d\bmh   \Bigg| + \Bigg|  \int_{\mathcal{H}_n}  \bmh\bmT_n^T \pi^{*}\left(\bmh|\bmY_{\infty},\bmxi_{\infty} \right) d\bmh  \Bigg|\\
		&\equiv& A_{n1}+A_{n2}+A_{n3}+A_{n4}+A_{n5}+A_{n6}.
	\end{eqnarray*}
	By Lemma~\ref{lm:conv_total_var}, $A_{n1}=o_p(1)$. Since $\pi^{*}\left(\bmh|\bmY_{\infty},\bmxi_{\infty} \right)$ is Gaussian and $P(||\bmh||\geq 1/6)=O(n^{-1/3})$ by the Chebyshev's inequality, we can obtain $A_{n2}=o_p(1)$. $\bmT_n=o_p(1)$ and the process of the proof in Lemma~\ref{lm:conv_total_var} lead to $A_{n3}=o_p(1)$. In addition, 
	by $\bmT_n=o_p(1)$ and the dominated convergence theorem, we obtain $A_{n4}=o_p(1)$. By the Cauchy-Schwarz inequality, we decompose $A_{n5}$  into 
	\begin{eqnarray*}
		A_{n5} &\leq& \frac{1}{2}\Bigg| \int_{\mathcal{H}_n}  \bmh\bmh^T \left\{\pi^{*}\left(\bmh|\bmY_{n},\bmxi_{n} \right)-  \pi^{*}\left(\bmh|\bmY_{\infty},\bmxi_{\infty} \right) \right\}        d\bmh   \Bigg| \\
		& & + \frac{1}{2}\Bigg| \int_{\mathcal{H}_n}  \bmT_n\bmT_n^T \left\{\pi^{*}\left(\bmh|\bmY_{n},\bmxi_{n} \right)-  \pi^{*}\left(\bmh|\bmY_{\infty},\bmxi_{\infty} \right) \right\}        d\bmh   \Bigg|. 
	\end{eqnarray*}
	By the results of $A_{n1}$ and $A_{n3}$, $A_{n5}=o_p(1)$. By the symmetric property of a Gaussian density, we have $\int_{\mathcal{H}_n} ||\bmh||\pi^{*}\left(\bmh|\bmY_{\infty},\bmxi_{\infty} \right) d\bmh=o(1) $ and $\bmT_n=o_p(1)$. Thus, we have $A_{n6}=o_p(1)$. Therefore, as $\nri$,
	\[
	\bmH_n(\hat{\bmtheta})^{-1} - \bmH_n(\bmtheta_0)^{-1}=o_p(1).
	\]
	Since the inversion on the cone of positive definite matrices is continuous, we conclude that as $\nri$
	\begin{equation}\label{eq:JhatJ}
		\bmH_n(\hat{\bmtheta}) \bmH_n(\bmtheta_0)^{-1} \CIP \bmI_q. 
	\end{equation}
	
	By Theorem~2, we have 
	\[
	\sqrt{n} \left( \hat{\bmtheta} -\bmtheta_0 \right) \CID N\left(\bmzero, \bmH_n(\bmtheta_0)^{-1}\bmJ_n(\bmtheta_0)\bmH_n(\bmtheta_0)^{-1} \right) \quad \text{as} \quad \nri. 
	\]
	Since the continuity of $g(\bmtheta)$ implies that $g(\hat{\bmtheta}) \CIP g(\bmtheta)$ and $\nabla g(\hat{\bmtheta}) \CIP g(\bmtheta_0)$ by Theorem~2, the delta method with $g(\bmtheta)$ leads to 
	\begin{equation}\label{eq:delta}
		\frac{\sqrt{n}\left\{g(\hat{\bmtheta}) - g(\bmtheta_0) \right\} }{\sqrt{\nabla g(\bmtheta_0)^T\bmH_n(\bmtheta_0)^{-1}\bmJ_n(\bmtheta_0)\bmH_n(\bmtheta_0)^{-1} \nabla g(\bmtheta)0) }} \CID N(0,1) \quad \text{as} \quad \nri.
	\end{equation}
	Thus, we also have 
	\begin{align*}
		\nabla
		&  g(\hat{\bmtheta})^T\hat{\bmH}_n(\bmtheta_0)^{-1}\bmJ_n(\hat{\bmtheta})\bmH_n(\hat{\bmtheta})^{-1} \nabla g(\hat{\bmtheta}) \\
		& \CIP \nabla g(\bmtheta_0)^T\bmH_n(\bmtheta_0)^{-1}\bmJ_n(\bmtheta_0)\bmH_n(\bmtheta_0)^{-1} \nabla g(\bmtheta_0) 
	\end{align*}
	as $\nri$. Therefore the Wald bound, 
	\[
	c_{g,n}(\alpha) = g(\hat{\bmtheta}) + q_{\alpha} \sqrt{\frac{g(\hat{\bmtheta})^T\bmH_n(\hat{\bmtheta})^{-1}\bmJ_n(\hat{\bmtheta})\bmH_n(\hat{\bmtheta})^{-1} \nabla g(\hat{\bmtheta})}{n}},
	\]
	satisfies 
	\[
	P\left\{ c_{g,n}(\alpha/2) \le g(\bmtheta_0) \le c_{g,n}(1-\alpha/2) \right\} \rightarrow 1-\alpha.
	\]
	by Equation \eqref{eq:delta} and the Slutsky's theorem. 
	
	The proof is complete.

	\section*{Lemmas }
	
	\setcounter{equation}{0}
	
	\begin{lemma}\label{lm01} Under Assumption A1, Assumption A2 holds if  the poputation criterion $M_n(\bmtheta)\equiv \E\left[ \xi_i\log f(\bmy_i|\bmtheta)\right] $ satisfies the following conditions:
		\begin{enumerate}
			\item[(i)] $M(\bmtheta)$ is continuous on $\Theta$,
			\item[(ii)] for any $\delta>0$, 
			\[
			\limsup_{n \rightarrow \infty} \Bigg[ \sup_{||\bmtheta-\bmtheta_0||>\delta} \Bigg\{M_n(\bmtheta)-M_n(\bmtheta_0)\Bigg\} \Bigg] <0, 
			\]
			\item[(iii)] $n^{-1}L_n(\bmtheta) - M_n(\bmtheta)$ converges to zero in probability uniformly over $\Theta$.
		\end{enumerate} 
	\end{lemma}
	
	\bigskip 
	
	\noindent \textbf{Proof)} We begin by decomposing the term as
	\begin{eqnarray}\label{eq:De_LLD}
		\frac{1}{n} \Bigg\{ L_n(\bmtheta)- L_n(\bmtheta_0) \Bigg\} &=& \Bigg[ \frac{1}{n} L_n(\bmtheta) - M_n(\bmtheta) \Bigg] - \Bigg[ \frac{1}{n} L_n(\bmtheta_0) - M_n(\bmtheta_0) \Bigg] \\ \nonumber 
		& & + \Bigg[ M_n(\bmtheta) - M_n(\bmtheta_0) \Bigg].
	\end{eqnarray}
	Condition (iii) implies that for any $\varepsilon' > 0$,
	\[
	\sup_{\bmtheta \in \Theta} \Bigg|\frac{1}{n}L_n(\bmtheta) - M_n(\bmtheta)\Bigg| = o_p(1).
	\]
	Therefore, with probability approaching 1, we can bound the first two bracketed terms of the right-hand side of Equation~\eqref{eq:De_LLD} as
	\[
	\Bigg| \Bigg[ \frac{1}{n} L_n(\bmtheta) - M_n(\bmtheta) \Bigg] - \Bigg[ \frac{1}{n} L_n(\bmtheta_0) - M_n(\bmtheta_0) \Bigg] \Bigg| \leq 2 \sup_{\bmtheta \in \Theta} \Bigg|\frac{1}{n}L_n(\bmtheta) - M_n(\bmtheta)\Bigg| < 2\varepsilon'.
	\]
	
	Condition (ii) states that for any $\delta > 0$, there exists a constant $\varepsilon'' > 0$ such that for sufficiently large $n$:
	\begin{equation}\label{eq:Mbound}
		\sup_{||\bmtheta-\bmtheta_0||>\delta} \Bigg\{M_n(\bmtheta)-M_n(\bmtheta_0)\Bigg\} \le - \varepsilon'',
	\end{equation}
	which shows that on the parameter space outside the $\delta$-ball around $\bmtheta_0$, the expected log-likelihood difference is uniformly negative and bounded away from zero.
	
	Hence, combining the bounds from Equation~\eqref{eq:Mbound}, for a given $\delta > 0$ and any $\varepsilon' > 0$, we have with probability approaching 1:
	\begin{align*}
		\sup_{||\bmtheta-\bmtheta_0||>\delta} \frac{1}{n} \Bigg\{ L_n(\bmtheta)- L_n(\bmtheta_0) \Bigg\} &\leq \sup_{||\bmtheta-\bmtheta_0||>\delta} \Bigg[ M_n(\bmtheta) - M_n(\bmtheta_0) \Bigg] + 2 \sup_{\bmtheta \in \Theta} \Bigg|\frac{1}{n}L_n(\bmtheta) - M_n(\bmtheta)\Bigg| \\
		&< -\varepsilon'' + 2\varepsilon'.
	\end{align*}
	By choosing $\varepsilon' = \varepsilon''/3$, the right-hand side becomes $-\varepsilon''/3$. Thus, for any $\delta>0$, there exists an $\varepsilon = \varepsilon''/3 > 0$ such that the average log-likelihood difference is strictly negative and bounded away from zero, with a probability that converges to 1. This is precisely the statement of Assumption A2.

	\newpage 
	
	\begin{lemma}\label{lm02}  Under Assumptions A1 and A2, Assumption A3 holds with 
		\[
		\bms_n(\bmtheta_0) = \nabla_{\bmtheta}L_n(\bmtheta_0) \quad \text{and} \quad \bmH_n(\bmtheta_0) =- \nabla_{\bmtheta\bmtheta^T}M_n(\bmtheta) = - \E\left[ \nabla_{\bmtheta\bmtheta^T}\xi_i\log f(\bmy_i|\bmtheta)\right]=O(1)
		\]
		if 
		\begin{enumerate} \itemsep 0.2cm
			\item[(i)] for some $\delta>0$, $L_n(\bmtheta)$ and $M_n(\bmtheta)$ are twice continuously differentiable in $\bmtheta$ when $||\bmtheta-\bmtheta_0||<\delta$, 
			\item[(ii)] there is $\bmJ_n(\bmtheta_0)=\E\left[\bms_n(\bmtheta)\bms_n(\bmtheta)^T \right]$ such that $\bmJ_n(\bmtheta_0)^{-1/2} \nabla_{\bmtheta} L_n(\bmtheta_0)/\sqrt{n} \stackrel{d}{\longrightarrow} N(0,\bmI_q)$, $\bmH_n(\bmtheta_0)=O(1)$ and $\bmJ_n(\bmtheta_0)=O(1)$ are uniforml positive definite, and
			\item[(iii)] for some $\delta>0$ and each $\varepsilon >0$ 
			\[
			\limsup_{n\rightarrow \infty} P^{*} \left\{ \sup_{||\bmtheta-\bmtheta_0||<\delta}  \Bigg| \frac{\nabla_{\bmtheta\bmtheta^T}L_n(\bmtheta)}{n} - \nabla_{\bmtheta\bmtheta^T}M_n(\bmtheta)  \Bigg|  \right\} = 0.   
			\]
		\end{enumerate} 
	\end{lemma}
	
	\bigskip 
	
	\noindent \textbf{Proof)} Condition (i) states that $L_n(\bmtheta)$ is twice continuously differentiable in a neighborhood of $\bmtheta_0$. By a second-order Taylor expansion around $\bmtheta_0$, we can write:
	\[
	L_n(\bmtheta) - L_n(\bmtheta_0) = (\bmtheta-\bmtheta_0)^T \nabla_{\bmtheta}L_n(\bmtheta_0) + \frac{1}{2}(\bmtheta-\bmtheta_0)^T \nabla_{\bmtheta\bmtheta^T}L_n(\bmtheta_{*}) (\bmtheta-\bmtheta_0)
	\]
	for some $\bmtheta_{*}$ between $\bmtheta$ and $\bmtheta_0$.
	Using the definitions $\bms_n(\bmtheta_0) = \nabla_{\bmtheta}L_n(\bmtheta_0)$ and $\bmH_n(\bmtheta_0) = - \nabla_{\bmtheta\bmtheta^T}M_n(\bmtheta_0)$, we can rearrange the second-order term to match the required form.
	\begin{eqnarray*}
		\frac{1}{2}(\bmtheta-\bmtheta_0)^T \nabla_{\bmtheta\bmtheta^T}L_n(\bmtheta_{*}) (\bmtheta-\bmtheta_0) &=& \frac{1}{2}(\bmtheta-\bmtheta_0)^T \left[ \nabla_{\bmtheta\bmtheta^T}L_n(\bmtheta_{*}) - n\nabla_{\bmtheta\bmtheta^T}M_n(\bmtheta_0) \right] (\bmtheta-\bmtheta_0) \\
		& & + \frac{1}{2}(\bmtheta-\bmtheta_0)^T n\nabla_{\bmtheta\bmtheta^T}M_n(\bmtheta_0) (\bmtheta-\bmtheta_0) \\
		&=&  - \frac{1}{2}(\bmtheta-\bmtheta_0)^T \left[ n\bmH_n(\bmtheta_0) \right] (\bmtheta-\bmtheta_0) + R_n(\bmtheta)
	\end{eqnarray*}
	where the remainder term is $R_n(\bmtheta) = \frac{1}{2}(\bmtheta-\bmtheta_0)^T \left[ \nabla_{\bmtheta\bmtheta^T}L_n(\bmtheta_{*}) - n\nabla_{\bmtheta\bmtheta^T}M_n(\bmtheta_0) \right] (\bmtheta-\bmtheta_0)$.
	This demonstrates that the formula of Assumption A3(i) holds.
	
	Assumptions A3(ii) and A3(iii) are directly provided by Condition (ii). The central limit theorem for the score function as
	\[
	\bmJ_n(\bmtheta_0)^{-1/2}\nabla_{\bmtheta} L_n(\bmtheta_0)/\sqrt{n} \stackrel{d}{\longrightarrow} N(0,\bmI_q)
	\]
	directly corresponds to Assumption A3(ii).
	$\bmH_n(\bmtheta_0)=O(1)$ and $\bmJ_n(\bmtheta_0)=O(1)$ are uniformly positive-definite constant matrices, which directly corresponds to Assumption A3(iii).
	
	Assumptions A3(i)(a) and A3(i)(b) require showing that $R_n(\bmtheta)$ from the Taylor expansion is negligible under certain conditions. The remainder term is defined as
	\[
	R_n(\bmtheta) = \frac{1}{2}(\bmtheta-\bmtheta_0)^T \left[ \nabla_{\bmtheta\bmtheta^T}L_n(\bmtheta_{*}) - n\nabla_{\bmtheta\bmtheta^T}M_n(\bmtheta_0) \right] (\bmtheta-\bmtheta_0).
	\]
	Let us analyze the term inside the bracket by applying Condition (iii), which states that for a small $\delta>0$ and each $\varepsilon>0$, $\sup_{||\bmtheta-\bmtheta_0||<\delta} \left| \frac{1}{n}\nabla_{\bmtheta\bmtheta^T}L_n(\bmtheta) - \nabla_{\bmtheta\bmtheta^T}M_n(\bmtheta) \right| = o_p(1)$.
	This implies that for any $\varepsilon'>0$,
	\[
	P^{*} \left\{ \sup_{||\bmtheta-\bmtheta_0||<\delta} \left| \frac{1}{n}\nabla_{\bmtheta\bmtheta^T}L_n(\bmtheta) - \nabla_{\bmtheta\bmtheta^T}M_n(\bmtheta) \right| > \varepsilon' \right\} \to 0 \quad \text{as } n\to\infty.
	\]
	By the continuity of the Hessian in Condition (i), we can also relate $\nabla_{\bmtheta\bmtheta^T}M_n(\bmtheta_*) - \nabla_{\bmtheta\bmtheta^T}M_n(\bmtheta_0)$ to be small for small $\delta$.
	\begin{align*}
		\frac{R_n(\bmtheta)}{n||\bmtheta-\bmtheta_0||^2} &= \frac{1}{2}\frac{(\bmtheta-\bmtheta_0)^T (\nabla_{\bmtheta\bmtheta^T}L_n(\bmtheta_{*}) - n\nabla_{\bmtheta\bmtheta^T}M_n(\bmtheta_0)) (\bmtheta-\bmtheta_0)}{n||\bmtheta-\bmtheta_0||^2} \\
		&= \frac{1}{2}\frac{(\bmtheta-\bmtheta_0)^T (\nabla_{\bmtheta\bmtheta^T}L_n(\bmtheta_{*}) - n\nabla_{\bmtheta\bmtheta^T}M_n(\bmtheta_{*}) ) (\bmtheta-\bmtheta_0)}{n||\bmtheta-\bmtheta_0||^2} \\
		&+\frac{1}{2}\frac{(\bmtheta-\bmtheta_0)^T ( n\nabla_{\bmtheta\bmtheta^T}M_n(\bmtheta_{*}) - n\nabla_{\bmtheta\bmtheta^T}M_n(\bmtheta_0)) (\bmtheta-\bmtheta_0)}{n||\bmtheta-\bmtheta_0||^2} \\
		&\leq \frac{1}{2}\left| \frac{\nabla_{\bmtheta\bmtheta^T}L_n(\bmtheta_{*})}{n} - \nabla_{\bmtheta\bmtheta^T}M_n(\bmtheta_{*}) \right| + \frac{1}{2}\Bigg| \nabla_{\bmtheta\bmtheta^T}M_n(\bmtheta_{*}) - \nabla_{\bmtheta\bmtheta^T}M_n(\bmtheta_0) \Bigg|.
	\end{align*}
	The first term on the right-hand side is $o_p(1)$ by Condition (iii). The second term goes to 0 as $||\bmtheta-\bmtheta_0|| \to 0$ due to the continuity of the Hessian of $M_n$ from Condition (i). This implies that, for a sufficiently small $\delta$, the second term can be made arbitrarily small. This, in conjunction with the result from Condition (iii), proves that the entire expression is $o_p(1)$. This is sufficient to prove both parts of Assumption A3(i). For Assumption A3(i)(a), since $|R_n(\bmtheta)| / ( n||\bmtheta-\bmtheta_0||^2 )$ is $o_p(1)$, for any $\varepsilon>0$, we can choose $\delta$ small enough and $n$ large enough so that the probability of this term being greater than $\varepsilon$ is less than $\varepsilon$. For the second part, note that $|R_n(\bmtheta)| \leq n||\bmtheta-\bmtheta_0||^2 \times o_p(1)$. On the set $||\bmtheta-\bmtheta_0||\leq M/\sqrt{n}$, this term is bounded by $n(M/\sqrt{n})^2 \times o_p(1) = M^2 \times o_p(1)$, which converges to 0 in probability, thus satisfying Assumption A3(i)(b).
	
	All four parts of Assumption A3 are thus shown to hold.


	\begin{lemma}\label{lm03} Suppose that Assumptions A1--A4 hold. Define 
		\begin{eqnarray}\label{eq:wn}
			w_n(\bmh) &=& \sum_{i\in\mathcal{S}} \xi_i \log f\left(\bmy_i\Bigg|\frac{\bmh}{\sqrt{n}}+\bmT_n\right) \\ \nonumber
			& & - \sum_{i\in\mathcal{S}} \xi_i \log f\left(\bmy_i|\bmtheta_0\right) -\frac{1}{2n}\bms_n(\bmtheta_0)^T\bmH_n(\bmtheta_0)^{-1}\bms_n(\bmtheta_0).
		\end{eqnarray}
		Then, for $\alpha \geq 0$, as $\nri$,
		\[
		\bmW_{n} \equiv \int_{\mathcal{H}_n} ||\bmh||^{\alpha} \Bigg|\exp\left(w_n(\bmh)\right)\pi\left(\frac{\bmh}{\sqrt{n}}+\bmT_n \right) - \exp\left(-\frac{1}{2}\bmh^T\bmH_n(\bmtheta_0)\bmh \right)\pi(\bmtheta_0)   \Bigg| d\bmh \CIP 0
		\]
		where $\mathcal{H}_n = \{ \sqrt{n}(\bmtheta - \bmtheta_0 ) - \bmH_n(\bmtheta_0)^{-1}\bms_n(\bmtheta_0)/\sqrt{n} : \bmtheta \in \Theta \}$, and $\bmT_n=\bmtheta_0+\bmU_n/\sqrt{n}$ with $\bmU_n=\bmH_n(\bmtheta_0)^{-1}\bms_n(\bmtheta_0)/\sqrt{n}$
	\end{lemma}
	\noindent \textbf{Proof)}  Using Assumption A3(i) to expand $w_n(\bmh)$ around $\bmtheta_0$ gives 
	\begin{eqnarray}\label{eq:wnr}
		w_n(\bmh)&=& \left(\frac{\bmh}{\sqrt{n}}+\bmT_n-\bmtheta_0 \right)^T\bms_n(\bmtheta_0) - \frac{1}{2} \left(\frac{\bmh}{\sqrt{n}}+\bmT_n-\bmtheta_0 \right)^T \left[n J(\bmtheta_0) \right]\left(\frac{\bmh}{\sqrt{n}}+\bmT_n-\bmtheta_0 \right) \nonumber\\
		& & - \frac{1}{2n} \bms_n(\bmtheta_0)^T \bmH_n(\bmtheta_0)^{-1}\bms_n(\bmtheta_0) + R_n\left(\frac{\bmh}{\sqrt{n}}+\bmT_n \right)\nonumber\\
		& =&  - \frac{1}{2}\bmh^T\bmH_n(\bmtheta_0)\bmh + R_n\left(\frac{\bmh}{\sqrt{n}}+\bmT_n \right)
	\end{eqnarray}
	with $\bmT_n-\bmtheta_0=\frac{1}{n} \bmH_n(\bmtheta_0)^{-1}\bms_n(\bmtheta_0)$. In addition, for $\delta>0$ and $0<M<\infty$ we can split $\bmW_n$ as
	\begin{eqnarray*}
		\bmW_n&=& \int_{||\bmh||\leq M} ||\bmh||^{\alpha} \Bigg|\exp\left(w_n(\bmh)\right)\pi\left(\frac{\bmh}{\sqrt{n}}+\bmT_n \right) - \exp\left(-\frac{1}{2}\bmh^T\bmH_n(\bmtheta_0)\bmh \right)\pi(\bmtheta_0)\Bigg| d\bmh \\
		& & + \int_{M\leq ||\bmh||\leq \delta\sqrt{n}} ||\bmh||^{\alpha} \Bigg|\exp\left(w_n(\bmh)\right)\pi\left(\frac{\bmh}{\sqrt{n}}+\bmT_n \right) - \exp\left(-\frac{1}{2}\bmh^T\bmH_n(\bmtheta_0)\bmh \right)\pi(\bmtheta_0)\Bigg| d\bmh \\
		& & + \int_{||\bmh||\geq \delta\sqrt{n}} ||\bmh||^{\alpha} \Bigg|\exp\left(w_n(\bmh)\right)\pi\left(\frac{\bmh}{\sqrt{n}}+\bmT_n \right) - \exp\left(-\frac{1}{2}\bmh^T\bmH_n(\bmtheta_0)\bmh \right)\pi(\bmtheta_0)\Bigg| d\bmh \\
		&\equiv& B_{1n} + B_{2n} +B_{3n}. 
	\end{eqnarray*}
	
	In order to show $B_{1n}=o_p(1)$, using Equation~\eqref{eq:wnr} we can derive 
	\begin{eqnarray*}
		B_{1n}&=& \sup_{||\bmh||\leq M} \left[ ||\bmh||^{\alpha} \Bigg|\exp\left( -\frac{1}{2}\bmh^T\bmH_n(\bmtheta_0)\bmh+R_n\left( \frac{\bmh}{\sqrt{n}}+\bmT_n \right) \right)\pi\left(\frac{\bmh}{\sqrt{n}}+\bmT_n \right) \right.\\
		& & \quad \quad \quad \quad    \left. - \exp\left(-\frac{1}{2}\bmh^T\bmH_n(\bmtheta_0)\bmh \right)\pi(\bmtheta_0)\Bigg| \right] \\
		&\leq & M^{\alpha} \sup_{||\bmh||\leq M}\exp\left( -\frac{1}{2}\bmh^T \bmH_n(\bmtheta_0)\bmh \right)  \left[ \sup_{||\bmh||\leq M}\exp\left(R_n\left( \frac{\bmh}{\sqrt{n}}+\bmT_n \right)\right) \right]  \\
		&  & \times \left[\sup_{||\bmh||\leq M} \Bigg|\pi \left( \frac{\bmh}{\sqrt{n}} +\bmT_n \right)-\pi(\bmtheta_0)  \Bigg|+\pi(\bmtheta_0)\sup_{||\bmh||\leq M}\Bigg| \exp\left(R_n\left( \frac{\bmh}{\sqrt{n}}+\bmT_n \right)\right) -1  \Bigg| \right].
	\end{eqnarray*}
	By Assumption A3(i)(b),  we have 
	\[
	\sup_{||\bmh||\leq M} \Bigg| R_n\left( \frac{\bmh}{\sqrt{n}}+\bmT_n \right) \Bigg|=o_p(1).
	\]
	Thus, we obtain 
	\[
	\sup_{||\bmh||\leq M} \Bigg| \exp\left(R_n\left( \frac{\bmh}{\sqrt{n}}+\bmT_n \right)\right) \Bigg|=O_p(1) \quad \& \quad  \sup_{||\bmh||\leq M} \Bigg| \exp\left(R_n\left( \frac{\bmh}{\sqrt{n}}+\bmT_n \right)\right)-1 \Bigg|=o_p(1).
	\]
	
	By Assumption A4, we can also derive the following term as
	\[
	\pi \left( \frac{\bmh}{\sqrt{n}} +\bmT_n \right)-\pi(\bmtheta_0) = \left( \frac{\bmh}{\sqrt{n}} +\bmT_n-\bmtheta_0 \right)^T \nabla_{\bmtheta}\pi(\bmtheta_0) + R\left(\left( \frac{\bmh}{\sqrt{n}} +\bmT_n-\bmtheta_0 \right)^2 \right).
	\] 
	Since $n^{-1/2} \bmH_n(\bmtheta_0)^{-1}\bms_n(\bmtheta_0)=o_p(1)$ by Assumptions A3(ii) and A3(iii), we have 
	\begin{eqnarray*}
		\sup_{||\bmh||\leq M} \Bigg| \pi \left( \frac{\bmh}{\sqrt{n}} +\bmT_n \right)-\pi(\bmtheta_0) \Bigg| &\leq & \sup_{||\bmh||\leq M} \Bigg|\Bigg| \frac{\bmh}{\sqrt{n}} +\frac{1}{n}\bmH_n(\bmtheta_0)^{-1}\bms_n(\bmtheta_0)\Bigg|\Bigg| ||\nabla_{\bmtheta}\pi(\bmtheta_0)|| \\
		& & + \sup_{||\bmh||\leq M}\Bigg| R\left( \left(\frac{\bmh}{\sqrt{n}} +\frac{1}{n}\bmH_n(\bmtheta_0)^{-1}\bms_n(\bmtheta_0) \right)^2   \right)     \Bigg|\\
		& =& O_p\left( \frac{1}{\sqrt{n}}\right).
	\end{eqnarray*}
	Hence, $B_{1n}=O_p(n^{-1/2})=o_p(1)$.
	
	Next, in order to show $B_{2n}=o_p(1)$,  first we have that 
	\[
	\pi(\bmtheta_0) \int_{M \leq ||\bmh ||< \delta \sqrt{n}} ||\bmh||^{\alpha} \exp \left( -\frac{1}{2}\bmh^T \bmH_n(\bmtheta_0)\bmh\right) d\bmh \CIP 0
	\]
	with large $M$ as $\nri$ by the properties of the prior function in Assumption A4 and multivariate normal distribution. Thus, it is enough to show, for each $\varepsilon>0 $, there exist large $M$ and small $\delta>0$ such that 
	\[
	\liminf_{\nri} P^{*} \left\{ \int_{M \leq ||\bmh|| < \delta \sqrt{n}} ||\bmh||^{\alpha} \Bigg| \exp(w_n(\bmh)) \pi \left(\bmT_n+\frac{\bmh}{\sqrt{n} } \right)  \Bigg| d\bmh \leq \varepsilon           \right\} \geq 1 - \varepsilon.
	\]
	In order to prove $B_{2n}=o_p(1)$, it is enough to show that for sufficiently large $M$ as $\nri$
	\begin{equation}\label{eq:wnpi}
		\exp(w_n(\bmh)) \pi \left(\bmT_n+\frac{\bmh}{\sqrt{n} } \right)  \leq C \exp \left(   
		-\frac{1}{4}\bmh^T \bmH_n(\bmtheta_0) \bmh \right)
	\end{equation}
	for some $C>0$ and all $M \leq ||\bmh || <\delta\sqrt{n}$. By the properties of the prior function in Assumption A4, $\pi \left(\bmT_n+\frac{\bmh}{\sqrt{n} } \right)$ is bounded over $M \leq ||\bmh || <\delta\sqrt{n}$, which means that we do not need to consider this proof. Note that by using Equation~\eqref{eq:wnr}, we have 
	\begin{eqnarray*}
		\exp\left(w_n(\bmh) \right) &=& \exp\left( -\frac{1}{2}\bmh^T \bmH_n(\bmtheta_0) \bmh +R_n\left(\frac{\bmh}{\sqrt{n}}+\bmT_n \right) \right)\\
		&\leq& \exp\left( -\frac{1}{2}\bmh^T \bmH_n(\bmtheta_0) \bmh +\Bigg| R_n\left(\frac{\bmh}{\sqrt{n}}+\bmT_n \right) \Bigg| \right).
	\end{eqnarray*}
	Since $|\bmT_n-\bmtheta_0|=O_p(n^{-1})=o_p(1)$, for any $\delta>0$
	\[
	\Bigg|\Bigg| \bmT_n + \frac{\bmh}{\sqrt{n}}-\bmtheta_0    \Bigg|\Bigg| \leq \Bigg|\Bigg| \bmT_n -\bmtheta_0    \Bigg|\Bigg| +  \Bigg|\Bigg| \frac{\bmh}{\sqrt{n}} \Bigg|\Bigg| \leq  o_p(1)+ 2 \delta
	\]
	for all $||\bmh||<\delta\sqrt{n}$. Thus, by Assumption A3(i)(a), with $\bmh=\sqrt{n} (\bmtheta-\bmtheta_0) -\frac{1}{\sqrt{n}}\bmH_n(\bmtheta_0)^{-1}\bms_n(\bmtheta_0)$ and $\bmtheta=\bmT_n+\bmh/\sqrt{n}$, there exists some small $\delta>0$ and large $M$ such that 
	\[
	\liminf_n P^{*}\left\{ \sup_{M\leq ||\bmh||< \delta\sqrt{n}} \frac{|R_n(\bmT_n+\bmh/\sqrt{n})|}{||\bmh +\bmH_n(\bmtheta_0)^{-1}\bms_n(\bmtheta_0)/\sqrt{n} ||^2} \leq \frac{1}{4}\min_{j=1,\dots,d} \lambda_j   \right\} \geq 1-\varepsilon, 
	\]
	where $\lambda_j$ is an eigenvalue of $\bmH_n(\bmtheta)$ for $i=1,\dots,d$. Since $||\bmH_n(\bmtheta_0)^{-1}\bms)_n(\bmtheta_0)/n||=o(1)$, for some $C>0$, we have 
	\begin{eqnarray}\label{eq:wne}
		& & \liminf_n P^{*} \left\{\exp(w_n(\bmh)) \leq C\exp \left(   -\frac{1}{4}\bmh^T \bmH_n(\bmtheta_0)\bmh\right)  \right\} \nonumber \\
		&\geq  & \liminf_n P^{*} \left\{\exp(w_n(\bmh)) \leq \exp \left(   -\frac{1}{2}\bmh^T \bmH_n(\bmtheta_0)\bmh + \frac{1}{4} \min_{j=1,\dots,d} \lambda_j ||\bmh||^2  \right) \right\} \nonumber  \\
		&\geq& 1- \varepsilon.
	\end{eqnarray}
	By the property of the prior function in A.4, $\pi(\cdot)$ is bounded for sufficiently large $M$ as $\nri$ and small $\delta>0$. Hence,  for sufficiently large $ M$
	as $\nri$, we show
	\[
	\exp(w_n(\bmh)) \leq C \exp \left(-\frac{1}{4}\bmh^T \bmH_n(\bmtheta_0) \bmh \right)
	\]
	for all $M \leq ||\bmh || < \delta \sqrt{n}$. In addition, the result implies  equation~\eqref{eq:wnpi}. Thus, we obtain $B_{2n}=o_p(1)$.
	
	Lastly, as we have the same method as the second proof, we start with 
	\[
	\pi(\bmtheta_0) \int_{||\bmh|| \geq \delta \sqrt{n}} ||\bmh||^{\alpha} \exp \left( -\frac{1}{2}\bmh^T \bmH_n(\bmtheta_0)\bmh\right) d\bmh \CIP 0
	\]
	Therefore, it is enough to show that as $\nri$
	\[
	\int_{||\bmh|| \geq \delta \sqrt{n}} ||\bmh||^{\alpha} \exp \left( w_n(\bmh)\right)\pi \left( \frac{\bmh}{\sqrt{n}} +\bmT_n\right) d\bmh \CIP 0
	\]
	By the definition of $\bmh$, we rewrite the above equation as
	\[
	n^{(\alpha+1)/2}\int_{||\bmtheta-\bmT_n||\geq \delta} ||\bmT_n-\bmtheta||^{\alpha} \exp \left( L_n(\bmtheta) - L_n(\bmtheta_0) - \frac{1}{2n} \bms_n(\bmtheta_0)^T \bmH_n(\bmtheta_0)^{-1}\bms_n(\bmtheta_0) \right) d\bmtheta.
	\]
	Since $\bmT_n -\bmtheta_0 \rightarrow 0$ as $\nri$, the above equation with probability 1 is bounded by 
	\[
	\kappa_n \cdot C \cdot n^{(\alpha+1)/2} \int_{||\bmtheta-\bmtheta_0||\geq \delta/2} \left( 1+||\bmtheta||^{\alpha} \right) \pi(\bmtheta) \exp\left(L_n(\bmtheta)-L_n(\bmtheta_0)\right)d\bmtheta,
	\]
	where 
	\[
	\kappa_n = \exp \left(- \frac{1}{2n} \bms_n(\bmtheta_0)^T \bmH_n(\bmtheta_0)^{-1}\bms_n(\bmtheta_0)  \right) =O(1). 
	\]
	By Assumption A2, there exists $\varepsilon>0$ such that 
	\[
	\liminf_{\nri} P^{*} \left\{   \sup_{||\bmtheta-\bmtheta_0||\geq \delta/2} \exp\left(L_n(\bmtheta)-L_n(\bmtheta_0)\right) \leq \exp(-n\varepsilon)\right\}=1. 
	\]
	Thus, the equation $\int_{||\bmh|| \geq \delta \sqrt{n}} ||\bmh||^{\alpha} \exp \left( w_n(\bmh)\right)\pi \left( \frac{\bmh}{\sqrt{n}} +\bmT_n\right)$ is bounded by 
	\[
	\kappa_n \cdot C \cdot n^{(\alpha+1)/2} \cdot \exp(-n\varepsilon ) \int_{\Theta} ||\bmtheta||^{\alpha} \pi(\bmtheta) d\bmtheta =o_p(1)
	\]
	because of the compactness of $\int_{\Theta} ||\bmtheta||^{\alpha}\pi(\bmtheta)d\bmtheta<\infty$. Hence, $B_{3n}=o_p(1)$. 
	
	Therefore, the proof is complete.


	\begin{lemma} \label{lm:conv_total_var}
		Suppose that Assumptions A1--A4 hold. As $n \rightarrow \infty$, 
		\[
		\int_{\mathcal{H}_n} ||\bmh||^2 \left| \pi^*(\bmh | \bmY_n,\bmxi_n) - \pi^*(\bmh | \bmY_\infty,\bmxi_\infty) \right| d \bmh \overset{p}{\rightarrow} 0,
		\]
		where $\mathcal{H}_n = \{ \sqrt{n}(\bmtheta - \bmtheta_0 ) - \bmH_n(\bmtheta_0)^{-1}\bms_n(\bmtheta_0)/\sqrt{n} : \bmtheta \in \Theta \}$ and 
		\[
		\pi^*(\bmh | \bmY_\infty,\bmxi_\infty) = \sqrt{\frac{\text{det} \bmH_n(\bmtheta_0)}{(2 \pi)^q }} \exp\left( -\frac{1}{2} \bmh^T \bmH_n(\bmtheta_0) \bmh \right).
		\]
	\end{lemma}
	
	\bigskip
	
	\noindent \textbf{Proof)} Suppose that $\bmH_n(\bmtheta)$ and $\bmJ_n(\bmtheta)$ do not depend on the size of $\mathcal{S}$ which is the sample size $n$ (i.e., $n=|\mathcal{S}|$). Let $\bmh = \sqrt{n}(\bmtheta-\bmT_n)$ where $\bmT_n=\bmtheta_0+\bmU_n/\sqrt{n}$ and $\bmU_n=\bmH_n(\bmtheta_0)^{-1}\bms_n(\bmtheta_0)/\sqrt{n}$. Then
	\begin{eqnarray*}
		\pi^*(\bmh|\bmY_n,\bmxi_n) &= & \frac{1}{\sqrt{n}} \pi\left( \frac{\bmh}{\sqrt{n}} +\bmtheta_0 + \frac{\bmU_n}{\sqrt{n}} \Bigg| \bmY_n,\bmxi_n\right) \\
		& = & \frac{\pi\left(\frac{\bmh}{\sqrt{n}}+\bmT_n  \right) \exp\left\{\sum_{i\in\mathcal{S}} \xi_i \log f\left(\bmy_i|\frac{\bmh}{\sqrt{n}}+\bmT_n\right) \right\} }{\int_{\mathcal{H}_n }\pi\left(\frac{\bmh}{\sqrt{n}}+\bmT_n  \right) \exp\left\{\sum_{i\in\mathcal{S}} \xi_i \log f\left(\bmy_i|\frac{\bmh}{\sqrt{n}}+\bmT_n\right) \right\}d\bmh} \\
		& =& \frac{1}{C_n}\pi\left(\frac{\bmh}{\sqrt{n}}+\bmT_n  \right) \exp\left\{w_n(\bmh) \right\},
	\end{eqnarray*}
	where $
	w_n(\bmh) = \sum_{i\in\mathcal{S}} \xi_i \log f\left(\bmy_i|\frac{\bmh}{\sqrt{n}}+\bmT_n\right) - \sum_{i\in\mathcal{S}} \xi_i \log f\left(\bmy_i|\bmtheta_0\right) -\frac{1}{2n}\bms_n(\bmtheta_0)^T\bmH_n(\bmtheta_0)^{-1}\bms_n(\bmtheta_0)
	$ and $C_n=\int_{\mathcal{H}_n }\pi\left(\frac{\bmh}{\sqrt{n}}+\bmT_n  \right) \exp\left\{w_n(\bmh) \right\}d\bmh$. 
	
	Thus, by the total variation of moment norms with $\alpha\geq 0$,
	\begin{eqnarray*}
		& & \int_{\mathcal{H}_n}||\bmh||^2 \Bigg| \frac{1}{C_n}\pi\left(\frac{\bmh}{\sqrt{n}}+\bmT_n  \right) \exp\left\{w_n(\bmh) \right\} - \sqrt{\frac{\text{det} \bmH_n(\bmtheta_0)}{(2 \pi)^q }} \exp\left( -\frac{1}{2} \bmh^T \bmH_n(\bmtheta_0) \bmh \right)\Bigg|d\bmh  \\
		&\leq & \frac{1}{C_n} \int_{\mathcal{H}_n}||\bmh||^2 \Bigg| \pi\left(\frac{\bmh}{\sqrt{n}}+\bmT_n  \right) \exp\left\{w_n(\bmh) \right\} - \exp\left( -\frac{1}{2} \bmh^T \bmH_n(\bmtheta_0) \bmh\right) \pi(\bmtheta_0)  \Bigg|d\bmh  \\
		& & + \frac{1}{C_n} \int_{\mathcal{H}_n}||\bmh||^2 \Bigg| C_n\sqrt{\frac{\text{det} \bmH_n(\bmtheta_0)}{(2 \pi)^q }}\exp\left( -\frac{1}{2} \bmh^T \bmH_n(\bmtheta_0)\right) - \exp\left( -\frac{1}{2} \bmh^T \bmH_n(\bmtheta_0) \bmh\right) \pi(\bmtheta_0)  \Bigg|d\bmh  \\
		& \equiv & \frac{1}{C_n}\left(A_{1n}+A_{2n}\right)
	\end{eqnarray*}
	
	By Lemma~\ref{lm03} with $\alpha=0$, as $\nri$,
	\[
	\int_{\mathcal{H}_n}  \Bigg|\exp\left(w_n(\bmh)\right)\pi\left(\frac{\bmh}{\sqrt{n}}+T_n \right) - \exp\left(-\frac{1}{2}\bmh^T\bmH_n(\bmtheta_0)\bmh \right)\pi(\bmtheta_0)   \Bigg| d\bmh \CIP 0.
	\]
	Thus, 
	\begin{eqnarray}\label{eq:Cn}
		C_n & \CIP & \int_{\mathcal{H}_n}\exp\left(-\frac{1}{2}\bmh^T\bmH_n(\bmtheta_0)\bmh \right)\pi(\bmtheta_0) d\bmh \nonumber \\
		& = & \pi(\bmtheta_0)\int_{\mathcal{H}_n}\exp\left(-\frac{1}{2}\bmh^T\bmH_n(\bmtheta_0)\bmh \right) d\bmh \nonumber\\
		&=& \pi(\bmtheta_0) \sqrt{\frac{(2\pi)^q}{|\det(\bmH_n(\bmtheta_0)|}} \nonumber\\
		& =& O(1)
	\end{eqnarray}
	Thus, $C_n=O_p(1)$ as $\nri$. Since $\pi(\bmtheta_0)\sqrt{\frac{(2\pi)^q}{|\det(\bmH_n(\bmtheta_0)|}}\neq 0$, by continuous mapping theorem $C_n^{-1} =O_p(1)$ as $\nri$. By Lemma~\ref{lm03} with $\alpha=2$, as $\nri$, we have 
	\[
	A_{1n}=o_p(1). 
	\]
	Since $\int_{\mathcal{H}_n}||\bmh||^2  \exp\left( -\frac{1}{2} \bmh^T \bmH_n(\bmtheta_0) \bmh\right) \pi(\bmtheta_0)d\bmh=O(1) $,  
	\begin{eqnarray*}
		A_{2n}&=& \Bigg|C_n\sqrt{\frac{\text{det} \bmH_n(\bmtheta_0)}{(2 \pi)^q }} -\pi(\bmtheta_0)\Bigg| \int_{\mathcal{H}_n}||\bmh||^2  \exp\left( -\frac{1}{2} \bmh^T \bmH_n(\bmtheta_0) \bmh\right) \pi(\bmtheta_0)d\bmh \\
		&=& \Bigg|C_n\sqrt{\frac{\text{det} \bmH_n(\bmtheta_0)}{(2 \pi)^q }} -\pi(\bmtheta_0)\Bigg| \times O(1).
	\end{eqnarray*}
	Equation \eqref{eq:Cn} implies that 
	\[
	\Bigg|C_n\sqrt{\frac{\text{det} \bmH_n(\bmtheta_0)}{(2 \pi)^q }} -\pi(\bmtheta_0)\Bigg|=o_p(1) \quad \text{ as} \quad \nri.
	\]
	Hence, by Slutsky's theorem, we have 
	\[
	\frac{1}{C_n}\left(A_{1n}+A_{2n}\right) =o_p(1) \quad \text{as} \quad \nri. 
	\]
	The proof is complete.

	\begin{lemma}\label{lm05} Suppose that Assumptions A1--A3 hold and the Bayes estimator is defined as
		\[
		\bmZ_n\equiv \arg\inf_{\bmz\in \mathbb{R}^q}Q_{\infty} (\bmz)
		\]
		where 
		\[
		Q_{\infty}(\bmz)= \int_{\bmz\in\mathbb{R}^q} (\bmz-\bmh-\bmU_n)^2\pi^{*}(\bmh|\bmY_{\infty},\bmxi_{\infty})d\bmh.
		\]
		Then, as $\nri$, 
		\[
		\bmZ_n-\sqrt{n}(\hat{\bmtheta}-\bmtheta_0) =o_p(1).
		\]
		with $\hat{\bmtheta}=\arg\sup_{\bmtheta\in\Theta} L(\bmtheta)$
	\end{lemma}
	
	\bigskip

	\noindent \textbf{Proof)} Since $\pi^{*}(\bmh|\bmY_{\infty},\bmxi_{\infty})$ is a centered Gaussian with precision $\bmH_n(\bmtheta_0)$,
	\[
	\E_{\pi^*}[\bmh]=0, \quad \text{and} \quad \int\|z-(\bmh+\bmU_n)\|^2\pi^*(\bmh) d\bmh = \|\bmz\|^2 - 2 \bmz^\top \bmU_n + C,
	\]
	where $C$ is a constant. Hence $Q_{\infty}(\bmz)$ is a strictly convex quadratic in $\bmz$ with gradient
	$\nabla_{\bmz} Q_\infty(\bmz)=2(\bmz-\bmU_n)$. Therefore, its unique minimizer is
	\begin{equation}\label{eq:Zn}
		\bmZ_n = \bmU_n =\bmH_n(\bmtheta_0)^{-1}\,\frac{\bms_n(\bmtheta_0)}{\sqrt{n}}.
	\end{equation}
	
	By Assumption A3, a second-order Taylor expansion of the score around $\bmtheta_0$ yields, uniformly on
	$\|\bmtheta-\bmtheta_0\|=O(n^{-1/2})$,
	\[
	0 =\bms_n(\hat{\bmtheta})
	=\bms_n(\bmtheta_0) - n \bmH_n(\bmtheta_0)\,(\hat{\bmtheta}-\bmtheta_0) + o_p(\sqrt{n}),
	\]
	where $\bmH_n(\bmtheta_0)$ is positive definite and $\bms_n(\bmtheta_0)/\sqrt{n}=O_p(1)$.
	Rearranging gives
	\begin{equation}\label{eq:Un}
		\sqrt{n}(\hat{\bmtheta}-\bmtheta_0)
		=
		\bmH_n(\bmtheta_0)^{-1}\frac{\bms_n(\bmtheta_0)}{\sqrt{n}} + o_p(1)
		=\bmU_n + o_p(1).
	\end{equation}
	
	From Equation~\eqref{eq:Zn}, $Z_n=U_n$, and from Equation~\eqref{eq:Un}, $\sqrt{n}(\hat{\bmtheta}-\bmtheta_0)=\bmU_n+o_p(1)$.
	Therefore, $\bmZ_n-\sqrt{n}(\hat{\bmtheta}-\bmtheta_0)=o_p(1)$, as claimed.


	\begin{lemma}\label{lm06} \citep{anderson1955inequaltiy}
		Let $f:\mathbb{R}^p \rightarrow [0,\infty)$ be symmetric with convex level sets. Let $\bmE \subset \mathbb{R}^p$ be a symmetric convex set. Then, for every $\bmy \in \mathbb{R}^p$ and every $\lambda\in[0,1]$, 
		\[
		\int_{\bmE} f(\bmx-\lambda\bmy)d\bmx \geq \int_{\bmE} f(\bmx-\bmy)d\bmx.
		\]
		In particular, the map $\bmy \mapsto \int_{\bmE} f(\bmx-\bmy)d\bmx$ is maximized at $\bmy=\bmzero$.
	\end{lemma}

	\bigskip

	\begin{lemma}\label{lm07} \citep{pollard1991asymptotic}
		Let $\{\lambda_n(\theta):\theta \in \Theta\}$ be a sequence of random convex functions defined on a convex, open subset $\Theta$ of $\mathbb{R}^q$. Suppose $\lambda(\cdot)$ is a real-valued function on $\Theta$ for which $\lambda_n(\theta) \rightarrow \lambda(\theta)$ in probability, for each $\theta$ in $\Theta$. Then, for each compact subset $K$ of $\Theta$, 
		\[
		\sup_{\theta\in K} |\lambda_n(\theta)-\lambda(\theta)| \rightarrow 0 \quad \text{in probability}.
		\]
		The function $\lambda(\cdot)$ is necessarily convex on $\Theta$.
	\end{lemma}

	
	\section*{Simulation Setting }
	
	\setcounter{equation}{0}
	
	This section describes the detailed setting of the simulation studies included in Section 4 of the main text. To generate the simulation data from $K=3$ mixture components of multivarite normal distributions with $ p=4$ variables, we used the following simulation parameters:
	$$
	\boldsymbol{\eta}  =  (0.25, 0.6, 0.15), \ 
	\bmmu_1 =  (2,2,3,2), \
	\bmmu_2  =  (6,5,6,6), \ \text{and} \
	\bmmu_3 =  (8,7,9,10) 
	$$  
	with following covariance matrices
	\begin{align*}
		\bmSigma_1 & = \left[ \begin{matrix}
			2 & 1.3 & 0.8 & 0.8 \\ 
			1.3 & 1 & 0.6 & 0.7 \\ 
			0.8 & 0.6 & 1 & 0.8 \\ 
			0.8 & 0.7 & 0.8 & 1 \\ 
		\end{matrix} \right], \bmSigma_2 = \left[ \begin{matrix}
			1.5 & -0.3 & 0.5 & 0.8 \\ 
			-0.3 & 1.5 & 0.6 & 0.7 \\ 
			0.5 & 0.6 & 1.2 & 0.6 \\ 
			0.8 & 0.7 & 0.6 & 1.0 \\ 
		\end{matrix} \right], 
		\bmSigma_3 = \left[ \begin{matrix}
			2 & 1.3 & 0.9 & 1.5 \\ 
			1.3 & 1.2 & 0.5 & 0.7 \\ 
			0.9 & 0.5 & 0.7 & 0.6 \\ 
			1.5 & 0.7 & 0.6 & 1.7 \\ 
		\end{matrix} \right].  
	\end{align*} 
	

\clearpage
\renewcommand{\refname}{References for Supplement}

	
	
	
	
	\end{document}